\documentclass[review]{elsarticle}

\usepackage[utf8]{inputenc}
\usepackage[T1]{fontenc}
\usepackage[english]{babel}
\usepackage{amsmath}
\usepackage{amsfonts}
\usepackage{amssymb}
\usepackage{enumitem}
\usepackage{booktabs}
\usepackage{xfrac}
\usepackage{csquotes}
\usepackage{tabularx}
\usepackage{graphicx}
\usepackage{multirow}
\usepackage[usenames,dvipsnames]{xcolor} 

\usepackage[normalem]{ulem} 

\usepackage{pgfplotstable}
\pgfplotsset{compat=1.17}
\usepackage{booktabs}
\usepackage{float}
\usepackage{changepage}
\usepackage{caption}
\usepackage[pdfencoding=auto, psdextra, hidelinks]{hyperref}
\allowdisplaybreaks 
\usepackage{amsthm}
\theoremstyle{definition}
\newtheorem{example}{Example}
\newtheorem{definition}{Definition}

\theoremstyle{plain}
\newtheorem{theorem}{Theorem}

\newtheorem{proposition}{Proposition}
\newtheorem{corollary}{Corollary}

\usepackage{setspace}
\biboptions{authoryear}

\begin{document}

\begin{frontmatter}

\title{A pessimist's approach to one-sided matching\\}

\author[KULeuvenAddress,UGentAddress]{Tom Demeulemeester\corref{mycorrespondingauthor}}
\ead{tom.demeulemeester@kuleuven.be}
\author[UGentAddress,FlandersMake]{Dries Goossens}
\author[KULeuvenAddress]{Ben Hermans}
\author[KULeuvenAddress]{Roel Leus}

\cortext[mycorrespondingauthor]{Corresponding author}

\address[KULeuvenAddress]{Research Center for Operations Research \& Statistics, KU Leuven, Belgium}
\address[UGentAddress]{Department of Business Informatics and Operations Management, Ghent University, Belgium}
\address[FlandersMake]{Corelab CVAMO, FlandersMake@UGent, Belgium}

\begin{abstract}
Inspired by real-world applications such as the assignment of pupils to schools or the allocation of social housing, the one-sided matching problem studies how a set of agents can be assigned to a set of objects when the agents have preferences over the objects, but not vice versa. For fairness reasons, most mechanisms use randomness, and therefore result in a probabilistic assignment. We study the problem of decomposing these probabilistic assignments into a weighted sum of ex-post \mbox{(Pareto-)efficient} matchings, while maximizing the worst-case number of assigned agents. This decomposition preserves all the assignments' desirable properties, most notably strategy-proofness. Next to discussing the complexity of the problem, we obtain tight lower and upper bounds on the optimal worst-case number of assigned agents. Moreover, we propose two alternative column generation frameworks for the introduced problem, which prove to be capable of finding decompositions with the theoretically best possible worst-case number of assigned agents, both for randomly generated data, and for real-world school choice data from the Belgian cities Antwerp and Ghent. Lastly, the proposed column generation frameworks are inherently flexible, and can therefore also be applied to settings where other ex-post criteria are desirable, or to find decompositions that satisfy other worst-case measures.
\end{abstract}


\end{frontmatter}

\section{Introduction}
\label{sec:intro}
We study assignment problems in which a set of \textit{agents} have preferences over a set of indivisible \textit{objects}, but not vice versa, while the agents can be assigned to at most one object. This problem models a wide range of applications, ranging from the assignment of pupils to schools to the allocation of social housing (see \cite{biro2017applications} and \cite{elster1992local} for an overview of applications). In this type of problem, which is generally referred to as bipartite matching with one-sided preferences or (Capacitated) House Allocation, fairness plays a crucial role, and monetary transfers are generally not allowed. Most mechanisms used in practice adopt randomness to ensure that all agents are treated equally. As a result, we obtain a \textit{probabilistic assignment} that determines for each agent-object pair the probability of being assigned to each other.

In order to implement this probabilistic assignment, however, we need to decompose it into a set of corresponding \textit{matchings}, which state for each agent-object pair whether they are assigned to each other or not. We will focus on two aspects of finding such a decomposition that have a great practical appeal to decision makers, but that have not yet been studied extensively.

First of all, we make the realistic assumption that decision makers are risk-averse. This implies that they pessimistically fear the worst-case outcome, and want to avoid ending up with an \enquote{unlucky} random draw that results, for example, in a matching that assigns exceptionally few agents to an object. We therefore analyze the problem of maximizing the worst-case number of assigned agents by all matchings in the decomposition of a probabilistic assignment. A second natural requirement is that all matchings in the decomposition of a probabilistic assignment be desirable on their own. A crucial property in the context of one-sided matching is \textit{ex-post (Pareto-)efficiency}: a matching is ex-post efficient if there is no other matching that is at least as preferred by all agents and strictly more preferred by at least one agent. If agents want to exchange their assigned objects, for example, this matching is not ex-post efficient.

Importantly, we do not want to harm the desirable properties, such as \textit{strategy-proofness}, that are related to the probabilistic assignment that is decomposed. We therefore focus on decomposing the original probabilistic assignment exactly, rather than to approximate it (contrary to, e.g., \cite{akbarpour2020approximate}, \cite{bronfman2018redesigning}).

Given a probabilistic assignment that satisfies certain fairness criteria, the well-known \textit{Birkhoff-von Neumann theorem} \citep{birkhoff1946, vonNeumann1953} guarantees that a decomposition over matchings exists. This decomposition, however, is generally not unique, and not all decompositions are equally desirable. To the best of our knowledge, the problem of selecting the most desirable among all feasible decompositions was only studied explicitly in two papers. First, \cite{budish2013} proposed a method to minimize the variance in the utility that is experienced by the agents between the different matchings in the decomposition. Nevertheless, their method is not capable of imposing certain desirable criteria, such as ex-post efficiency, upon the matchings in the decomposition. Second, \cite{kesten2017efficient} require the matchings to be ex-post efficient, but their algorithm does not allow the decision-maker to choose among the possibly many decompositions based on some worst-case measure or on alternative ex-post properties. The frameworks that we propose overcome both limitations.
        
Lastly, designated algorithms have been proposed to obtain decompositions with attractive ex-post properties for specific probabilistic assignments. Most notably, the decomposition of the assignment by the \textit{Probabilistic Serial} (PS) mechanism \citep{bogomolnaiaMoulin2001} has been studied by \cite{freeman2020best} and by \cite{aziz2020simultaneously}. Whereas their algorithms can only be applied to find a decomposition for one specific probabilistic assignment, our solution methods can be applied to any probabilistic assignment, as discussed in more detail in Section 5.4.

We introduce two new problems: (1) the problem of finding a decomposition of a probabilistic assignment~$X$ that maximizes the worst-case number of assigned agents, which we refer to as MD$(X)$, and (2) the same problem with the additional requirement that all matchings in the decomposition be ex-post efficient, which we refer to as MD-SD$(X)$.  

Our contributions are both theoretical and practical. From a theoretical perspective, we propose an algorithm to solve MD$(X)$ in polynomial time, and prove that the optimal value of MD$(X)$ will always be equal to the expected number of assigned agents by $X$, rounded down. When solving MD-SD$(X)$, however, we find that the same algorithm can only be used for a specific class of probabilistic assignments, which includes the well-studied Probabilistic Serial \citep{bogomolnaiaMoulin2001} mechanism, but not the often used \textit{Random Serial Dictatorship} (RSD)~\citep{abdul1998}. For general probabilistic assignments, MD-SD$(X)$ becomes $\mathcal{NP}$-hard. Specifically for the decomposition of the probabilistic assignment by the RSD mechanism, we provide tight lower and upper bounds on the optimal increase in the worst-case number of assigned agents. Lastly, we propose two column generation frameworks to find an optimal decomposition for MD-SD$(X)$, and we discuss how they can be adapted for finding decompositions that optimize alternative worst-case criteria, or that contain matchings satisfying different ex-post properties.

From a practical perspective, we have conducted computational experiments to assess the performance of the proposed methods, and we show that the realized gain in the worst-case number of assigned agents, in comparison to RSD, supports the adoption of our methods in practical applications. Because, to the best of our knowledge, no set of benchmark instances exists for one-sided matching, we have developed a parameterized data generation tool. Next to evaluating the proposed methods on generated instances of different sizes, we also assessed their performance on two real-world school choice instances from the Belgian cities Antwerp and Ghent.

The remainder of this paper is structured as follows. Section~\ref{sec:setup} formally defines one-sided matching. Next, Section~\ref{sec:maximin} introduces the problems MD$(X)$ and MD-SD$(X)$, whereas the complexity and the optimal values of both problems are discussed in Section~\ref{sec:analysis}. Specifically for MD-SD$(X)$, Section~\ref{sec:column_gen} describes the column generation frameworks, while Section~\ref{sec:results} analyzes the performance of the proposed methods on existing and randomly generated one-sided matching instances. Lastly, Section~\ref{sec:conclusion} concludes.

\section{Preliminaries}
\label{sec:setup}
A \textit{one-sided matching problem} is defined by a four-tuple $(N, O, >, q)$, where $N$ is a finite set of agents, and $O$ is a finite set of objects to which the agents in~$N$ want to be assigned. The preference profile $> \;= (>_1, \ldots, >_{|N|})$ contains for each agent~$i \in N$ a strict ordering $>_i$ over the objects in $O \cup \varnothing$, where $\varnothing$ refers to the outside option of not being assigned to any object. We write $j >_i k$ if agent~$i \in N$ prefers object~$j$ to object~$k$, and $\varnothing >_i l$ if she prefers the outside option over object~$l$, with $j, k, l \in O$. Moreover, the capacities $q = (q_1, \ldots, q_{|O|})$ determine for each object~$j \in O$ the maximum number of agents $q_j \in \mathbb{N}$ that can be assigned to it. Denote the set of all one-sided matching instances by $\mathcal{I}$.

A \textit{matching} $M = [m_{ij}]$ is a binary matrix, indexed by all agents and objects, in which $m_{ij} = 1$ if $M$ assigns agent $i \in N$ to object $j \in O$, and $m_{ij} = 0$ otherwise. An alternative way to indicate that agent $i \in N$ is assigned to object $j \in O$ is to state that $M(i) = j$. Similarly, a \textit{probabilistic assignment} $X = [x_{ij}]$, with $x_{ij} \in [0,1]$, can be interpreted as the probabilities with which the agents in~$N$ are assigned to the objects in $O$. A probabilistic assignment is \textit{feasible} if the sum of the allocation probabilities of each agent is  at most one, and the sum of the allocation probabilities of each object does not exceed the object's capacity, i.e.,
\vspace{-0.05cm}
\begin{enumerate}[label=(\roman*)]
	\itemsep 0em
	\item $\sum_{j \in O} x_{ij} \leq 1, \quad \forall \: i \in N$,
	\item $\sum_{i \in N} x_{ij} \leq q_j, \quad \forall \: j \in O$.
\end{enumerate}
Denote the set of all feasible matchings by $\mathcal{M}$, and the set of all feasible probabilistic assignments by $\Delta\mathcal{M}$, with $\mathcal{M} \subseteq \Delta\mathcal{M}$. Moreover, denote the expected number of assigned agents in a probabilistic assignment $X$ by $\mu(X) = \sum_{(i,j) \in N \times O}x_{ij}$. Note that the number of agents that is assigned to an object in matching~$M \in \mathcal{M}$ equals~$\mu(M)$. 

\subsection{Implementing probabilistic assignments}
In most practical applications, one does not make a decision by using the probabilistic assignment directly. Instead, the assignment is typically implemented by a lottery over matchings, in which the weights are chosen such that the expected result is equivalent to the initial probabilistic assignment. A natural question that arises is over which sets of feasible matchings a probabilistic assignment can be implemented in this way. \cite{budish2013} formalized this using the following definition:

\begin{definition}
	\label{def:impl}
	A feasible probabilistic assignment $X \in \Delta \mathcal{M}$ is \emph{implementable over}~$\mathcal{M}'$, where $\mathcal{M}' = \{M^t\}_{t=1}^T \subseteq \mathcal{M}$ is a subset of $T$ feasible matchings, if there exist non-negative numbers $\lambda = \{\lambda^t\}_{t=1}^T$, with $\sum_{t=1}^T \lambda^t = 1$, such that
	\begin{equation*}
	X = \sum_{t=1}^T \lambda^t M^t.
	\end{equation*}
\end{definition}

Given an arbitrary feasible probabilistic assignment $X \in \Delta \mathcal{M}$, the well-known \textit{Birkhoff-von Neumann theorem} \citep{birkhoff1946, budish2013, vonNeumann1953} ensures that there will always exist a subset of feasible matchings over which $X$ can be implemented. This means that each feasible probabilistic assignment $X$ can be written as a convex combination of a set of feasible matchings $\{M^t\}_{t=1}^T \subseteq \mathcal{M}$. In order to implement~$X$, it suffices to randomize over the matchings in this convex combination in such a way that matching $M^t$ is selected as the final matching with a probability equal to its weight $\lambda^t$ in the convex combination. The following definition formalizes the difference between implementing and decomposing a probabilistic assignment:
\begin{definition}
    \label{def:impl_decomp}
    Given a probabilistic assignment $X \in \Delta \mathcal{M}$ that is implementable over $\mathcal{M}' = \{M^t\}_{t=1}^T \subseteq \mathcal{M}$ with weights $\lambda = \{\lambda^t\}_{t=1}^T$,  
    \begin{enumerate}[label=(\roman*), leftmargin=*,labelindent=0pt]
    	\itemsep 0em
        \item an \textit{implementation} of $X$ is an algorithm that randomly selects a single matching $M^t \in \mathcal{M}^\prime$ according to the probability distribution defined by the weights $\lambda$;
        \item a \textit{decomposition} of $X$ is the tuple $(\mathcal{M}', \lambda)$.
    \end{enumerate}
\end{definition}

In other words, an implementation of a probabilistic assignment is an algorithm that implicitly describes an underlying decomposition by randomly selecting each of the matchings with a probability that is equal to the corresponding weight in the decomposition. In most practical applications, therefore, only an implementation of a probabilistic assignment is needed to make a decision, and not its entire decomposition. Nevertheless, having the entire decomposition increases the transparency of the procedure. To illustrate the introduced terminology, consider the following example.

\begin{example} 
	\label{ex:intro}
	Let $N = \{1,2,3,4\}$ be a set of agents and $O = \{a,b,c\}$ be a set of objects. The capacities of the objects are equal to $q = (2,1,1)$. Consider the following preference lists $>_i$ of the agents $i \in N$ over the objects~$j \in O$, and consider the feasible probabilistic assignment~$X^1$ (in which the agents correspond to the rows and the objects to the columns):
	
	\begin{minipage}{0.55 \textwidth}
	\begin{align*}
	\text{agents } 1, 2: & \quad a > b > c > \varnothing\\
	\text{agents } 3, 4: & \quad a > \varnothing > b > c\\
	&
	\end{align*}
	\end{minipage}
	\begin{minipage}{0.45 \textwidth}
	$X^1 = \begin{pmatrix}
	\sfrac{1}{2} & \sfrac{5}{12} & \sfrac{1}{12}\\
	\sfrac{1}{2} & \sfrac{5}{12} & \sfrac{1}{12}\\
	\sfrac{1}{2} & 0 & 0\\
	\sfrac{1}{2} & 0 & 0\\
	\end{pmatrix}$
	\end{minipage}
	
	\noindent Agents~1 and 2 prefer object~$a$ to object~$b$, and object~$b$ to object~$c$, while agents~3 and 4 only prefer object $a$ to the outside option of not being assigned to any object. Moreover, agent~1 has a probability of $\frac{5}{12}$ of being assigned to object~$b$ in $X^1$. We can rewrite $X^1$ in the following way:
	
	\begin{equation*}
	X^1 = \frac{5}{12} \begin{pmatrix}
	0 & 1 & 0\\
	1 & 0 & 0\\
	1 & 0 & 0\\
	0 & 0 & 0\\
	\end{pmatrix}
	+ \frac{5}{12} \begin{pmatrix}
	1 & 0 & 0\\
	0 & 1 & 0\\
	0 & 0 & 0\\
	1 & 0 & 0\\
	\end{pmatrix}
	+ \frac{1}{12} \begin{pmatrix}
	1 & 0 & 0\\
	0 & 0 & 1\\
	1 & 0 & 0\\
	0 & 0 & 0\\
	\end{pmatrix}
	+ \frac{1}{12} \begin{pmatrix}
	0 & 0 & 1\\
	1 & 0 & 0\\
	0 & 0 & 0\\
	1 & 0 & 0\\
	\end{pmatrix}.
	\end{equation*}
	Denote the four matchings in the decomposition of $X^1$ by $M^1, \:M^2, \: M^3$ and $M^4$. Thus, probabilistic assignment $X^1$ is implementable over $\{M^t\}_{t=1}^4$. Observe that $\mu(X^1) = 3$ and that every matching~$M^t$ also assigns~$\mu(M^t) = 3$ agents to an object, for $t=1,\ldots, 4$.
\end{example}

\subsection{Mechanisms and their properties}
In one-sided matching, objects are indifferent about which agents will be assigned to them. Therefore, there may exist ties between agents that have to be broken in order to obtain a matching. Define a \textit{tie-breaking rule}~$\sigma = [\sigma_j]$ for $j \in O$, where $\sigma_j = (\sigma_j(1), \ldots, \sigma_j(|N|))$ is a strict ordering over the agents in $N$. When object~$j$'s capacity is insufficient, its tie-breaking rule $\sigma_j$ is used to decide which agent should be assigned to object~$j$, among all agents who prefer that object equally. Denote the set of all tie-breaking rules by $\Sigma$.

A \textit{deterministic mechanism} $\pi:\mathcal{I} \times \Sigma \mapsto \mathcal{M}$ is a function that returns a feasible matching~$M^{\pi(I, \sigma)} \in \mathcal{M}$ for each one-sided matching instance $I \in \mathcal{I}$, and for each tie-breaking rule $\sigma \in \Sigma$. Similarly, a \textit{probabilistic mechanism} $\psi: \mathcal{I} \mapsto \Delta\mathcal{M}$ returns a feasible probabilistic assignment~$X^{\psi(I)} \in \Delta \mathcal{M}$ for each one-sided matching instance~$I \in \mathcal{I}$. Note that a probabilistic mechanism only returns a probabilistic assignment, and not necessarily its implementation or decomposition. Unless stated differently, the term \enquote{mechanism} will refer to probabilistic mechanisms in the remainder of this paper. Moreover, we will simply refer to $X^{\psi(I)}$ by $X^\psi$ when the instance~$I$ is clear from the context. 

When designing a mechanism, there are several desirable properties one wants to satisfy. An optimality concept that has received broad attention in the context of one-sided matching (e.g., \cite{abdul1998}, \cite{abraham2004pareto}) is \textit{ex-post (Pareto-)efficiency}. A matching~$M\in \mathcal{M}$ is ex-post efficient if there is no other matching~$M'\in \mathcal{M}$ that is at least as preferred as $M$ by all agents and strictly more preferred than~$M$ by at least one agent. The \textit{Serial Dictatorship} (SD) mechanism is a surprisingly simple deterministic mechanism to obtain an ex-post efficient matching. Denote the set of all tie-breaking rules with identical orderings for all objects by $\Sigma' \subseteq \Sigma$. Given a one-sided matching instance~$I \in \mathcal{I}$ and a strict ordering $\sigma \in \Sigma'$ of the agents, the SD mechanism will first assign the first-ranked agent~$\sigma(1)$ to her most preferred object, then the second-ranked agent~$\sigma(2)$ to her most preferred object among the objects with remaining capacity, etc. The resulting matching will be SD$(I, \sigma) = M^{\text{SD}(I, \sigma)}$. The SD mechanism possesses several desirable properties beside ex-post efficiency \citep{saban2015complexity}: it is \textit{neutral} (invariant to relabeling the objects), \textit{nonbossy} (no agent can change another agent's allocation without changing her own), \textit{strategy-proof} (submitting true preferences is optimal for all agents), and easy to compute. Moreover, given a one-sided matching instance $I \in \mathcal{I}$, the SD mechanism can generate any ex-post efficient matching in $I$ \citep{abdul1998, abraham2004pareto}. We will therefore denote the set of all ex-post efficient matchings by $\mathcal{M}^{\text{SD}} \subseteq \mathcal{M}$. 

Despite its desirable properties, the SD mechanism is rarely used in real-world applications, because it is not \textit{anonymous}: two agents with the same preference list might be treated differently, simply because one of them appears before the other one in the ordering of the agents that serves as an input for SD\@. In order to overcome this issue, \cite{abdul1998} introduced the \textit{Random Serial Dictatorship} (RSD) mechanism, which will randomly select an ordering to which the SD mechanism is applied. Applying the RSD mechanism to a one-sided matching instance $I = (N,O,>,q) \in \mathcal{I}$ will result in a probabilistic assignment~$X^{\text{RSD}(I)}$, which is the equally weighted sum of the matchings obtained by the SD mechanism for all different orderings $\sigma \in \Sigma'$:
\begin{equation}
\label{eq:RSD}
X^{\text{RSD}(I)} = \text{RSD}(I) = \frac{1}{|N|!} \sum_{\sigma \in \Sigma'} \text{SD}(I,\sigma).
\end{equation}
In the remainder of the paper, we will simply refer to $X^{\text{RSD}(I)}$ as $X^{\text{RSD}}$ if the instance is clear from the context. Note that the definition of the RSD mechanism is directly linked to an implementation of $X^{\text{RSD}}$ in the sense of Definition~\ref{def:impl_decomp}. In the remainder of this paper, we will use the term \textit{RSD mechanism} to refer to the function that maps an instance~$I$ to the probabilistic assignment $X^{\text{RSD}(I)}$, and the term \textit{RSD algorithm} to refer to the implementation of $X^{\text{RSD}(I)}$ which results in the matching SD$(I,\sigma)$ by randomly selecting a tie-breaking rule $\sigma \in \Sigma'$.

Although the RSD algorithm is widely used for real-world one-sided matching problems, \cite{bogomolnaiaMoulin2001} showed that it fails to satisfy the following notion of efficiency that is stronger than ex-post efficiency.

\begin{definition}
	\label{def:ex-antePE}
	A probabilistic assignment~$X\in \Delta \mathcal{M}$ is \emph{ordinally efficient} if and only if there is no probabilistic assignment $X' \in \Delta \mathcal{M}$ that is at least as preferred as $X$ by all agents and strictly more preferred than $X$ by at least one agent.
\end{definition}

Contrary to the RSD mechanism, the \textit{Probabilistic Serial} (PS) mechanism, that was developed by \cite{bogomolnaiaMoulin2001}, will always result in an ordinally efficient assignment. To obtain this probabilistic assignment, assume that time $t$ runs continuously from 0 to 1. At each point in time, each agent \enquote{eats} with a uniform eating speed of one from her most preferred object that has remaining capacity. At time $t=1$, we obtain the ordinally efficient probabilistic assignment~$X^{\text{PS}}$. Note that an implementation of $X^{\text{PS}}$ is not specified by its construction, in contrast to $X^{\text{RSD}}$.

Next to resulting in an ordinally efficient probabilistic assignment, the PS mechanism satisfies several other desirable properties \citep{bogomolnaiaMoulin2001}. Similarly to the RSD mechanism, PS will be anonymous. Moreover, PS will be \textit{envy-free} (no agent will prefer the assignment probabilities of another agent, for any compatible utility function), while RSD is not. In contrast to the RSD mechanism, however, PS will not be strategy-proof. In fact, \cite{bogomolnaiaMoulin2001} showed that no mechanism can satisfy ordinal efficiency, strategy-proofness and anonymity at the same time. More generally, the extent to which probabilistic assignments can satisfy certain combinations of properties has been widely studied in the literature, leading to an interesting series of impossibility results \citep{athanassoglou2011house,bogomolnaiaMoulin2001, martini2016strategy, nesterov2017fairness, ramezanian2022robust, zhou1990conjecture}.

One crucial observation in the context of this paper is that many relevant properties such as ordinal efficiency, anonymity, and strategy-proofness are inherent to the probabilistic assignment~$X \in \Delta \mathcal{M}$. Hence, all such properties are retained regardless of the decomposition of~$X$. Other properties, such as ex-post efficiency, are defined on the level of the matchings in the decomposition of~$X$. To ensure that a decomposition of~$X$ satisfies a property of the latter category, we thus have to explicitly enforce this property on all matchings in the decomposition.

\section{Maximin decomposition}
\label{sec:maximin}
Although the Birkhoff-von Neumann theorem is an important result, it does not guarantee desirable outcomes in practical applications. From the perspective of an individual agent, a decomposition might be undesirable because one of the matchings in the decomposition is undesirable in itself (Example~\ref{ex:BvN_notPE}). Decision makers might also prefer one decomposition to another because they are risk-averse (Example~\ref{ex:maximin}).

\begin{example}
	\label{ex:BvN_notPE} 
	Reconsider the decomposition of probabilistic assignment~$X^1$ in Example~\ref{ex:intro}. The third matching in the decomposition, $M^3$, is clearly sub-optimal as it is not ex-post efficient. We know that agent~2 prefers object~$b$ to object~$c$. Nevertheless, she is assigned to object~$c$ while object~$b$ has unused capacity. A similar argument holds for agent~1 in $M^4$.
\end{example}

\begin{example}
	\label{ex:maximin}	
	Consider an instance with four agents and two objects. The capacities of the objects are equal to $q=(2,2)$, and the agents' preferences are as indicated below. $X^2$ is a feasible probabilistic assignment for this instance.\\
	\begin{minipage}{0.6 \textwidth}
		\begin{align*}
		\text{agents } 1, 2: & \quad a > b > \varnothing\\
		\text{agents } 3, 4: & \quad a > \varnothing > b\\
		&
		\end{align*}
	\end{minipage}
	\begin{minipage}{0.4 \textwidth}
		$X^2 = \begin{pmatrix}
		\sfrac{1}{2} & \sfrac{1}{2} \\
		\sfrac{1}{2} & \sfrac{1}{2} \\
		\sfrac{1}{2} & 0 \\
		\sfrac{1}{2} & 0 \\
		\end{pmatrix}$
	\end{minipage}

	\noindent It can be easily verified that the following two decompositions of $X^2$ are both feasible:
	\begin{align*}
	X^2 &= \frac{1}{2} \begin{pmatrix}
	0 & 1\\
	0 & 1\\
	1 & 0\\
	1 & 0\\
	\end{pmatrix}
	+ \frac{1}{2} \begin{pmatrix}
	1 & 0\\
	1 & 0\\
	0 & 0\\
	0 & 0\\
	\end{pmatrix} = \frac{1}{2} \begin{pmatrix}
	0 & 1\\
	1 & 0\\
	1 & 0\\
	0 & 0\\
	\end{pmatrix}
	+ \frac{1}{2} \begin{pmatrix}
	1 & 0\\
	0 & 1\\
	0 & 0\\
	1 & 0\\
	\end{pmatrix}\label{eq:decomp2}.
	\end{align*}
	Although both decompositions result in the same probabilistic assignment $X^2$, a risk-averse decision maker will probably prefer the second decomposition over the first one: the second decomposition of $X^2$ is guaranteed to always assign three agents to an object, whereas there is a risk that only two agents will be assigned to an object by the first decomposition of $X^2$.
\end{example}

Inspired by Examples~\ref{ex:BvN_notPE} and \ref{ex:maximin}, we identify two criteria that impact the desirability of a decomposition. First of all, the agents generally have a set of properties that, depending on the context, they want to see satisfied by all matchings in the decomposition (Example~\ref{ex:BvN_notPE}). These properties could include, for example, ex-post efficiency, or group-specific quota. Secondly, we argue that risk-averse decision makers want to minimize the risk of ending up with an undesirable matching when implementing a probabilistic assignment. What exactly is considered to be an undesirable matching depends on the application at hand, but examples include matchings with a low average utility of the agents, a high average preference for the received object, a high total travel distance \cite[e.g., ][]{biro2021complexity}, a low \textit{popularity} \citep[e.g., ][]{kondratev2022minimal, mccutchen2008least}, or a low total number of assigned agents (Example~\ref{ex:maximin}). In this paper, we focus on maximizing the worst-case number of assigned agents. As will be discussed in Section~\ref{subsec:flexibility}, however, the column generation framework that we introduce can be easily adapted to find a decomposition that optimizes other worst-case criteria, or can even be adapted to different settings in which other criteria than ex-post efficiency should be satisfied by the matchings in the decomposition.

More specifically, we study the two following, closely related problems. First, consider the problem of maximizing the worst-case number of assigned agents in the decomposition of a probabilistic assignment~$X\in \Delta\mathcal{M}$, while allowing all feasible matchings in~$\mathcal{M}$ to be used in this decomposition. We will refer to this problem as \textsc{Maximin Decomposition of} $X$, or simply MD$(X)$. Denoting all matchings in $\mathcal{M}$ that assign at least $k \in \mathbb{N}$ agents to an object by $\mathcal{M}_k \subseteq \mathcal{M}$, we formally define MD$(X)$ as follows:
\begin{definition}
    \label{def:MD(X)}
	\textsc{Maximin Decomposition of} ${X}$ (MD$(X)$). Given a one-sided matching instance $I=(N, O, >, q) \in \mathcal{I}$, and a probabilistic assignment $X \in \Delta \mathcal{M}$, find a decomposition of $X$ over $\mathcal{M}_k$ that maximizes $k$.
\end{definition}

Second, we define a problem similar to MD$(X)$, but with the additional requirement that all matchings in the decomposition of $X$ be ex-post efficient. Denote the set of all ex-post efficient matchings that assign at least $k \in \mathbb{N}$ agents to an object by $\mathcal{M}^{\text{SD}}_k \subseteq \mathcal{M}^{\text{SD}}$.

\begin{definition}
	\label{def:MD-SD(X)}
	\textsc{Maximin Decomposition of $X$ over $\mathcal{M}^{\text{SD}}$} (MD-SD$(X)$). Given a one-sided matching instance $I=(N, O, >, q) \in \mathcal{I}$, and a probabilistic assignment $X \in \Delta \mathcal{M}$, find a decomposition of $X$ over $\mathcal{M}^{\text{SD}}_k$ that maximizes $k$.
\end{definition}

In the remainder of this paper, we will refer to the largest value of $k$ for which~$X$ is implementable over $\mathcal{M}_k$, resp.\
$\mathcal{M}_k^{\text{SD}}$, as the \textit{optimal value} of MD$(X)$, resp.\ MD-SD$(X)$, and we will denote the optimal value of MD-SD$(X)$ by $z(X)$. 

Denote the floor-operator by $\lfloor \cdot \rfloor$, and the ceiling-operator by $\lceil \cdot \rceil$. We can identify the following intuitive upper bound on the optimal values of MD$(X)$ and MD-SD$(X)$, which is based on the observation that it is not possible to decompose a probabilistic assignment~$X$ by only using matchings that assign strictly more agents to an object than the expected number of assigned agents in $X$.

\begin{proposition}
	\label{prop:UB_MD}
	The optimal values of MD$(X)$ and MD-SD$(X)$ are upper bounded by~$\lfloor \mu(X)\rfloor$, for all $X \in \Delta \mathcal{M}$.
\end{proposition}

\section{Obtaining a maximin decomposition: theoretical results}
\label{sec:analysis}
In this section, we examine the complexity and the optimal values of MD$(X)$ and MD-SD$(X)$. We show in Section~\ref{subsec:MD(X)} that the optimal value of MD$(X)$ is $\lfloor\mu(X)\rfloor$, for all $X \in \Delta \mathcal{M}$, and we propose a polynomial-time algorithm to find a corresponding decomposition. In Section~\ref{subsec:complex_MD-SD}, we establish that for MD-SD$(X)$ the same result only holds for a specific class of probabilistic matchings, whereas MD-SD$(X)$ is $\mathcal{NP}$-hard for general probabilistic assignments. Specifically for MD-SD$(X^\text{RSD})$, in Section~\ref{subsec:bounding_MD-SD(X^RSD)} we provide a lower and upper bound on the optimal gain in the worst-case number of assigned agents compared to the RSD algorithm. 

\subsection{Complexity of \texorpdfstring{MD$(X)$}{MD(X)}}
\label{subsec:MD(X)}
The following theorem states that obtaining a decomposition of~$X$ that attains the upper bound~$\lfloor \mu(X)\rfloor$ of Proposition~\ref{prop:UB_MD} can be achieved in polynomial time. The existence of such a decomposition follows from a result by \cite{budish2013} for a more general constraint structure known as a \textit{bihierarchy}. Budish et al.\ further describe an implementation in the sense of Definition~\ref{def:impl_decomp} that runs in polynomial time. The main novelty of our result is that it generalizes the method of Budish et al.\ in the spirit of the Birkhoff-von Neumann algorithm such that it returns a complete decomposition in polynomial time, and not simply one matching in this decomposition.

\begin{theorem}
	\label{theorem:impl_variab}
	For every probabilistic assignment $X \in \Delta \mathcal{M}$, we can find a decomposition~$(\mathcal{M}^\prime, \lambda)$ of~$X$ in polynomial time in which each matching $M \in \mathcal{M}^\prime$ assigns either $\lfloor \mu(X)\rfloor$ or $\lceil \mu(X) \rceil$ agents.
\end{theorem}

\begin{proof}
	Consider a probabilistic assignment $X \in \Delta \mathcal{M}$ and assume that $\mu(X) \in \mathbb{N}$. We will describe an iterative approach to obtain a decomposition of~$X$ in which all matchings assign exactly~$\mu(X)$ agents. The case where $\mu(X)$ is fractional can be dealt with by adding an additional row and column to~$X$, corresponding to a dummy agent~$a$ and a unit-capacity object~$o$, where the unique non-zero element~$(a,o)$ has an assignment probability equal to $\lceil \mu(X) \rceil - \mu(X)$. A decomposition that always assigns~$\lceil \mu(X) \rceil$ agents in this modified instance, then guarantees to assign either $\lfloor \mu(X)\rfloor$ or $\lceil \mu(X) \rceil$ non-dummy agents in the original instance.
	
	Following \cite{budish2013}, we can rewrite the conditions for~$X$ to be a feasible probabilistic assignment using a \emph{constraint structure}~$\mathcal{H} \subset 2^{N \times O}$, where~$\mathcal{H}$ contains a \textit{constraint set} $S \subseteq N \times O$ with associated \emph{quota}~$q_S$ for each agent, for each object, and for each agent-object pair. In particular, define $$\mathcal{H} = \bigcup_{(i,j) \in N \times O}\{(i,j)\} \cup \bigcup_{i \in N} (\{i\} \times O) \cup \bigcup_{j \in O} (N \times \{j\}),$$ and let~$q_S = q_j$ if $S = N \times \{j\}$ for some $j \in O$, and $q_S = 1$ for all other~$S\in\mathcal{H}$. A probabilistic assignment~$X$ is then feasible if and only if $0 \leq \sum_{(i,j) \in S} x_{ij} \leq q_S$ for every~$S \in \mathcal{H}$. Observe that~$\vert \mathcal{H} \vert = \vert N \vert \cdot \vert O \vert + \vert N \vert  + \vert O \vert = \mathcal{O}(\vert N \vert \cdot \vert O \vert)$.
	
	Each iteration~$t$ of our iterative approach starts from a given probabilistic assignment~$X^t$, where~$X^t = X$ if~$t = 1$, and $X^t$ is defined in the previous iteration otherwise. If~$X^t$ is a matching, i.e., if $X^t \in \mathcal{M}$, then the procedure terminates and we define~$M^t = X^t$. Otherwise, we apply to~$X^t$ a modified version of the algorithm by \cite{budish2013}, which we describe in \ref{app:budish_modified}, in order to obtain a matching $M^t \in \mathcal{M}$ with the following two properties:
	\vspace{-0.1cm}
	\begin{enumerate}[label=(\roman*)]
		\itemsep 0em
		\item $\mu(M^t) = \mu(X^t)$; \label{prop_cardinality}
		\item if $\sum_{(i,j) \in S}x^t_{ij}$ is integer, then $\sum_{(i,j) \in S}m^t_{ij} = \sum_{(i,j) \in S}x^t_{ij}$, for each $S \in \mathcal{H}$. \label{prop_integrality}
	\end{enumerate}
	\vspace{-0.35cm}
	Next, let~$\lambda^t$ be the largest value for $\lambda \in \mathbb{R}$ such that
	$X^t + \lambda(X^t - M^t) \in \Delta\mathcal{M}$, where~$\lambda^t$ can be determined by checking which of the constraint sets~$S\in \mathcal{H}$ becomes binding first. The procedure then defines $X^{t + 1} = X^t + \lambda^t(X^t - M^t)$,
	and proceeds to the next iteration. 
	
	Now assume that our procedure terminates in~$T$ iterations. It follows by expanding the above recursion that the matchings~$\{M^t\}_{t=1}^T$ and weights~$\hat{\lambda} = \{\hat{\lambda}^t\}_{t=1}^T$, with $$\hat{\lambda}^t = \begin{cases}
		\frac{\lambda^t}{1 + \lambda^t} & \text{if $t = 1$,}\\
		\frac{\lambda^t}{1 + \lambda^t} \prod_{u = 1}^{t-1} \frac{1}{1 + \lambda^u}& \text{if $t \in \{2, \ldots, T-1\}$,}\\
		\prod_{u = 1}^{t-1} \frac{1}{1 + \lambda^u} & \text{if $t = T$,}
	\end{cases}$$ yield a decomposition of~$X$ in which all matchings assign~$\mu(X)$ agents. This latter guarantee follows from Property~\ref{prop_cardinality} of~$M^t$ and by definition of~$X^{t + 1}$ for each $t = 1, \ldots, T - 1$.
	
	We claim that our procedure terminates in~$T \leq \vert \mathcal{H} \vert$ iterations. For each probabilistic assignment $X^\prime \in \Delta\mathcal{M}$, denote by $$\textstyle\tau(X^\prime) = \vert \{S \in \mathcal{H}\colon \sum_{(i,j) \in S}x^\prime_{ij} \in \mathbb{N}\} \vert$$ the number of constraint sets~$S \in \mathcal{H}$ that are integer in~$X^\prime$. Observe that~$\tau(X^\prime) \leq \vert \mathcal{H} \vert$, where equality is reached if and only if~$X^\prime \in \mathcal{M}$. To prove our claim, it thus suffices to show that, for every iteration~$t$ with~$X^t \notin \mathcal{M}$, it holds that $\tau(X^{t + 1}) > \tau(X^t)$.
	
	Consider an arbitrary iteration~$t$ in which~$X^t \notin \mathcal{M}$, and denote~$R^t = X^t - M^t$. Observe that~$R^t \neq 0$ since~$X^t \neq M^t$. By Property~\ref{prop_integrality} of~$M^t$, it holds for every $S \in \mathcal{H}$ for which $\sum_{(i,j) \in S} x^t_{ij}$ is integer that $\sum_{(i,j) \in S} r^t_{ij} = 0$ and thus, by definition of~$X^{t+1}$, that $\sum_{(i,j) \in S} x^{t+1}_{ij} =  \sum_{(i,j) \in S} x^{t}_{ij}$. Property~\ref{prop_integrality} of~$M^t$ also implies that $\sum_{(i,j) \in S} x^t_{ij}$ is fractional for every $S \in \mathcal{H}$ with $\sum_{(i,j) \in S} r^t_{ij} \neq 0$, and thus there exists a~$\lambda > 0$ for which $X^t + \lambda R^t \in \Delta\mathcal{M}$. The maximal choice~$\lambda^t$ for~$\lambda$ is then such that there is a $S^\prime \in \mathcal{H}$ with $\sum_{(i,j) \in S} x^{t}_{ij}$ fractional and $\sum_{(i,j) \in S} x^{t+1}_{ij}$ integer (because a new constraint becomes binding). This yields $\tau(X^{t + 1}) > \tau(X^t)$, which proves our claim. 
	
	Since the modified algorithm of \cite{budish2013}, described in \ref{app:budish_modified}, runs in time $\mathcal{O}(\vert \mathcal{H} \vert \min\{\vert N \vert, \vert O \vert\})$ and since finding~$\lambda^t$ in a given iteration~$t < T$ can be done in time~$\mathcal{O}(\vert \mathcal{H}\vert)$, we can obtain our desired decomposition of~$X$ in time  $\mathcal{O}(\vert \mathcal{H} \vert^2 \min\{\vert N \vert, \vert O \vert\})$.
\end{proof}

Theorem~\ref{theorem:impl_variab} implies that the optimal value of MD$(X)$ is equal to $\lfloor \mu(X)\rfloor$, for any $X \in \Delta \mathcal{M}$, and that the corresponding decomposition can be found in polynomial time. In other words, if the decision maker does not require specific properties to be satisfied by all matchings in the decomposition of a probabilistic assignment $X$, it is always possible to decompose $X$ by only using matchings that assign at least $\lfloor \mu(X) \rfloor$ agents to an object.

\subsection{Complexity of \texorpdfstring{MD-SD($X$)}{MD-SD(X)}}
\label{subsec:complex_MD-SD}
In contrast to MD$(X)$, the complexity of MD-SD$(X)$ depends on the properties of the probabilistic assignment~$X$. Consider the following property.

\begin{definition}
\label{def:robust}
	\citep{aziz2015ex} A probabilistic assignment~$X \in \Delta\mathcal{M}$ is \emph{robust ex-post efficient} if and only if every possible decomposition of $X$ only contains ex-post efficient matchings.
\end{definition}

Denote the set of robust ex-post efficient probabilistic assignments by $\Delta \mathcal{M}^{\text{R}} \subseteq \Delta \mathcal{M}$. Define a probabilistic assignment $X \in \Delta\mathcal{M}$ to be \textit{ex-post efficient} if a decomposition of $X$ exists that only contains ex-post efficient matchings. The following proposition shows the relationship between the three efficiency concepts for probabilistic assignments that have been mentioned in this paper.

\begin{proposition}
	\label{prop:eff_implications}
	\citep{aziz2015ex}  Ordinal efficiency of a probabilistic assignment implies robust ex-post efficiency, which in turn implies ex-post efficiency.
\end{proposition}

Although the above result was originally established for a setting with an equal number of agents and unit-capacity objects  \citep[][Theorem~5]{Aziz2014structurearXiv}, the proof directly carries over to our setting. Proposition~\ref{prop:eff_implications} has the following implications for the PS mechanism.

\begin{corollary}
	\label{cor:PS_robust}
	The PS mechanism will always return a probabilistic assignment that is  robust ex-post efficient.
\end{corollary}

Moreover, with respect to the RSD mechanism, \cite{aziz2015ex} obtained the following result (see Example~\ref{ex:BvN_notPE}, where $X^1 = X^{\text{RSD}}$, and $M^3, M^4 \notin \mathcal{M}^{\text{SD}}$). 

\begin{proposition}
	\label{prop:RSD_robust}
	\citep{aziz2015ex} The RSD mechanism may return a probabilistic assignment that is not robust ex-post efficient.
\end{proposition}

Using these results, we can now identify an important class of probabilistic assignments for which we obtain the same results as in Theorem~\ref{theorem:impl_variab}, i.e., for which the optimal value of MD-SD$(X)$ is guaranteed to be equal to its upper bound $\lfloor \mu(X)\rfloor$, and for which we can find a decomposition in polynomial time.

\begin{theorem}
	\label{theor:MD-P}
	For every robust ex-post efficient probabilistic assignment~$X \in \Delta \mathcal{M}^{\emph{R}}$, we can find a decomposition $(\mathcal{M}', \lambda)$ of $X$ in polynomial time in which each matching~$M \in \mathcal{M}'$ assigns either $\lfloor \mu(X)\rfloor$ of $\lceil \mu(X)\rceil$ agents.
\end{theorem}
\begin{proof}
	Consider an arbitrary robust ex-post efficient probabilistic assignment $X \in \Delta \mathcal{M}^{\text{R}}$. By using the polynomial-time algorithm from the proof of Theorem~\ref{theorem:impl_variab}, we can find a decomposition of $X$ in which all matchings assign either $\lfloor\mu(X)\rfloor$ or $\lceil \mu(X)\rceil$ agents to an object. By definition of robust ex-post efficiency, all matchings in this decomposition will be ex-post efficient. 
\end{proof}

Combining this result with Corollary~\ref{cor:PS_robust}, we can conclude that we can decompose the probabilistic assignment~$X^{\text{PS}}$ in polynomial time over a set of ex-post efficient matchings that assign at least $\lfloor \mu(X^{\text{PS}})\rfloor$ agents to an object. From Proposition~\ref{prop:RSD_robust}, however, we know that we cannot guarantee the ex-post efficiency of the matchings in the decomposition if we apply the same approach to~$X^{\text{RSD}}$.

Although Theorem~\ref{theor:MD-P} showed that MD-SD$(X)$ is polynomially solvable when~$X$ is robust ex-post efficient, the following result holds for general probabilistic assignments.

\begin{theorem}
	\label{theor:NP-hard_MD}
	MD-SD$(X)$  is $\mathcal{NP}$-hard in general.
\end{theorem}

\begin{proof}
The result immediately follows from the fact that it is already $\mathcal{NP}$-complete to determine whether an arbitrary probabilistic assignment can be implemented over ex-post efficient matchings when objects have unit capacities \citep{aziz2015ex}. 
\end{proof}

Note that the proof of Theorem~\ref{theor:NP-hard_MD} does not apply for $X^{\text{RSD}}$, because $X^{\text{RSD}}$ is guaranteed to have at least one decomposition over ex-post efficient matchings, given by Equation~(\ref{eq:RSD}). Hence, the complexity status of MD-SD$(X^{\text{RSD}})$ is open.

\subsection{Bounds on the optimal value of \texorpdfstring{MD-SD($X^{\text{RSD}}$)}{MD-SD(X-RSD)}}
\label{subsec:bounding_MD-SD(X^RSD)}
By Theorem~\ref{theor:MD-P}, we know that the optimal value~$z(X)$ of MD-SD$(X)$ equals its upper bound~$\lfloor\mu(X)\rfloor$ when $X$ is robust ex-post efficient. In this section, we analyze the possible range of values of $z(X)$ when $X \notin \Delta \mathcal{M}^{\text{R}}$. More specifically, we focus on the RSD mechanism, given its practical and theoretical importance: we establish that $\frac{1}{2}\lfloor \mu(X^{\emph{RSD}})\rfloor$ is an asymptotically tight lower bound on $z(X^{\text{RSD}})$, whereas twice the worst-case number of assigned agents by any ex-post efficient matching is an asymptotically tight upper bound on $z(X^{\text{RSD}})$.

Denoting the cardinality of a minimum-cardinality ex-post efficient matching in a one-sided matching instance~$I \in \mathcal{I}$ by $p^-(I)$, which is equal to the worst-case number of assigned agents by the RSD algorithm, we can identify the following bounds on $z(X^{\text{RSD}(I)})$.

\begin{proposition}
    \label{prop:LB-RSD_1}
    $\frac{1}{2}\lfloor \mu(X^{\emph{RSD}(I)})\rfloor < z(X^{\emph{RSD}(I)}) < 2p^-(I)$, for all $I \in \mathcal{I}$.
\end{proposition}

\begin{proof}
    By definition of MD-SD$(X^{\text{RSD}(I)})$, and given the upper bound on $z(X^{\text{RSD}(I)})$ from Proposition~\ref{prop:UB_MD}, we know that $p^-(I) \leq z(X^{\text{RSD}(I)}) \leq \lfloor \mu(X^{\text{RSD}(I)})\rfloor$. Hence, what remains to be shown is that $\frac{1}{2}\lfloor \mu(X^{\text{RSD}(I)})\rfloor < p^-(I)$. 
    
    Denote by~$p^+(I)$ the cardinality of a maximum-cardinality ex-post efficient matching in a one-sided matching instance~$I \in \mathcal{I}$, and note that $p^-(I) \leq p^+(I)$. We will consider two cases. First, if $p^-(I) = p^+(I)$, clearly $\frac{1}{2} \lfloor\mu(X^{\text{RSD}(I)}) \rfloor < p^-(I)$. Second, if $p^-(I) < p^+(I)$, we know that $\lfloor \mu(X^{\text{RSD}(I)}) \rfloor < p^+(I) \leq 2p^-(I)$, where the first inequality follows from Equation~(\ref{eq:RSD}), and the second inequality is a result by \cite{abraham2004pareto}.
\end{proof}

In the following two theorems, the proofs of which are included in Appendices \hyperref[appendix:proofThLB]{B} and \hyperref[appendix:proofThUB]{C}, we show that for each of the two strict inequalities in Proposition~\ref{prop:LB-RSD_1}, a family of instances exists for which both sides of the inequality converge. First of all, the following result states that there exists a family of instances for which $z(X^{\text{RSD}})$ cannot achieve any improvement in the worst-case number of assigned agents compared to the RSD algorithm, causing the ratio of $z(X^{\text{RSD}})$ over the upper bound $\lfloor \mu(X^{\text{RSD}})\rfloor$ to converge to $\frac{1}{2}$.

\begin{theorem}
    \label{theorem:LB-RSD_2}
    There exists a family of instances $I_k \in \mathcal{I}$, with $k\geq 2$ an integer, for which 
    \begin{equation*}
        \lim_{k \to \infty} {\frac{z(X^{\emph{RSD}(I_k)})}{\lfloor \mu(X^{\emph{RSD}(I_k)})\rfloor}}=\frac{1}{2}.
    \end{equation*}
\end{theorem}

Secondly, we show that there exist instances where $z(X^{\text{RSD}})$ realizes the maximum possible improvement in the worst-case number of assigned agents with respect to the RSD algorithm, as the ratio of $z(X^{\text{RSD}})$ over $p^-(I)$ converges to 2.

\begin{theorem}
    \label{theorem:UB-RSD_2}
    There exists a family of instances $I_\ell \in \mathcal{I}$, with $\ell\geq 2$ an integer, for which 
    \begin{equation*}
        \lim_{\ell \to \infty} {\frac{z(X^{\emph{RSD}(I_\ell)})}{p^-(I_\ell)}}=2.
    \end{equation*}
\end{theorem}

Lastly, we provide some intuition about the relation between the expected number of assigned agents by the RSD mechanism and the maximum number of assigned agents by any ex-post efficient matching in an instance $I\in\mathcal{I}$, which we denote by $p^+(I)$. Following a similar reasoning as in the proof of Proposition~\ref{prop:LB-RSD_1}, we know that when $p^-(I) < p^+(I)$ it must hold that $1<\frac{p^+(I)}{\mu(X^{\text{RSD}(I)})}<2$. First, the family of instances~$I_k$ as defined in the proof of Theorem~\ref{theorem:LB-RSD_2} shows that this intuitive lower bower bound is asymptotically tight. Second, this intuitive upper bound can be strengthened for larger instances, as discussed by \cite{bhalgat2011social} and \cite{krysta2019size} for RSD, and by \cite{bogomolnaia2015size} for PS.

\section{Column generation}
\label{sec:column_gen}
In order to be able to solve MD-SD$(X)$ to optimality for reasonably large instances when $X$ is not robust ex-post efficient, we introduce two alternative \textit{column generation} frameworks. Recall that the total number of tie-breaking rules $|\Sigma'|$ that can be used as an input for the SD mechanism to obtain an ex-post efficient matching is equal to $|N|!$, which implies that the number of ex-post efficient matchings $|\mathcal{M}^{\text{SD}}|$ can be exponential in the number of agents. A linear programming model that determines the weights of each matching in the decomposition would therefore contain an exponential number of decision variables. A common approach to tackle this issue is \textit{column generation} (e.g., \cite{bertsimas1997introduction}). The idea is to first solve the model with a subset of the variables in the \textit{restricted master problem}. Subsequently, we check whether this leads to an optimal solution over all variables. If this is not the case, a new variable will be added to the model, which is then solved again until the restricted master problem returns an optimal solution over all variables. New variables are generated by a separate subproblem called the \textit{pricing problem}, which identifies a variable with a negative \textit{reduced cost} or shows that no such variable exists.

In the remainder of this section, we propose two alternative column generation frameworks that will tackle the problem of finding an optimal solution to MD-SD$(X)$ from different angles. In both frameworks, we solve MD-SD$(X)$ by means of a binary search for $z(X)$, in which each step consists of solving a linear program to check whether a probabilistic assignment $X$ is implementable over $\mathcal{M}^{\text{SD}}_k$ for some integer~$k$. In Sections~\ref{subsec:RMP} and~\ref{subsec:pricing_problem}, by only using matchings in~$\mathcal{M}^{\text{SD}}_k$, we will minimize the largest difference between~$X$ and the resulting probabilistic assignment. Alternatively, in Section~\ref{subsec:alternative_formulation}, we enforce the resulting decomposition to be exactly equal to~$X$ while only using ex-post efficient matchings, and we will maximize the cumulative weights of the matchings in~$\mathcal{M}^{\text{SD}}_k$ in this decomposition. Although both frameworks are similar, their optimal solutions differ when~$X$ is not implementable over~$\mathcal{M}^{\text{SD}}_k$ for some integer~$k$. The performance and the interpretations of both frameworks will be discussed in Section~\ref{subsec:observ}.

\subsection{Minimizing the deviation from~\texorpdfstring{$X$}{X}: restricted master problem}
\label{subsec:RMP}
In this section, we formulate the restricted master problem $\left[\text{RMP}\right]$ for a column generation framework that aims to find a decomposition of a probabilistic assignment~$X$ in which all matchings are ex-post efficient and assign at least~$k$ agents to an object, for a given value of~$k$ in the binary search. To guide this search, we will minimize the largest difference over all agent-object pairs between~$X$ and the resulting probabilistic assignment. Let $X = [x_{ij}]$, with $(i,j) \in N\times O$, be the probabilistic assignment we want to decompose, and let $\tilde{\mathcal{M}}^{\text{SD}}_k \subseteq \mathcal{M}^{\text{SD}}_k$ be the subset of the variables that are considered in the restricted master problem. Moreover, define decision variables $\lambda^{t}$ to be the weight of matching $M^{t}\in 
\tilde{\mathcal{M}}^{\text{SD}}_k$ in the decomposition of~$X$, and let $s$ be an auxiliary variable.
\begin{subequations}
	\begin{align}
	&\left[\text{RMP}\right] & &\min &s &&\\ 
	& & &\: \text{s.t.} &\sum_{t:M^t \in \tilde{\mathcal{M}}^{\text{SD}}_k} \lambda^t m_{ij}^{t} - x_{ij} &\geq 0  &\forall \; (i,j) \in N \times O,\label{con:rmp2_1}\\
	& & & &\sum_{t:M^t \in \tilde{\mathcal{M}}^{\text{SD}}_k} \lambda^t m_{ij}^{t} - x_{ij} &\leq s  &\forall \; (i,j) \in N \times O, \label{con:rmp2_2}\\
	& & & &\sum_{t:M^t \in \tilde{\mathcal{M}}^{\text{SD}}_k} \lambda^t &= 1,&\label{con:rmp2_3}\\
	& & & &\lambda^{t} &\geq 0  &\forall \; t: M^t \in \tilde{\mathcal{M}}^{\text{SD}}_k. \label{con:rmp2_4}
	\end{align}
\end{subequations}

Constraints~(\ref{con:rmp2_1}) impose that the probability of agent~$i$ being assigned to object~$j$ is at least as high in the solution as in the original probabilistic assignment~$X$. Constraints~(\ref{con:rmp2_2}) ensure, together with the objective function, that~$s$ will be equal to the maximum difference between the constructed probabilistic assignment and $X$, over all agent-object pairs. As a result, a feasible decomposition of~$X$ over $\tilde{\mathcal{M}}^{\text{SD}}_k$ has been found if the objective value of $\left[\text{RMP}\right]$ equals zero. Furthermore, constraints~(\ref{con:rmp2_3}) and (\ref{con:rmp2_4}) state that the weights in the decomposition are non-negative and sum to one. To ensure the feasibility of $\left[\text{RMP}\right]$ for any subset~$\tilde{\mathcal{M}}^{\text{SD}}_k$, we can include an artificial \textit{super-column}~$M^0$ in which $m^0_{ij} = 1$ for all agent-object pairs $(i,j) \in N \times O$. Moreover, to reduce the number of constraints, we can remove redundant constraints~(\ref{con:rmp2_1}) when $x_{ij} = 0$, and constraints~(\ref{con:rmp2_2}) when $x_{ij} = 1$.

Denote the dual variables related to the constraints (\ref{con:rmp2_1}), (\ref{con:rmp2_2}), and (\ref{con:rmp2_3}) by $u_{ij} \geq 0$, $v_{ij} \leq 0$, and $w$, respectively, for each $(i,j) \in N \times O$. Then a solution of $\left[\text{RMP}\right]$ with dual values $u^* = [u^{*}_{ij}]$, $v^* = [v^{*}_{ij}]$, and $w^*$ is optimal over all variables in $\mathcal{M}_k^{\text{SD}}$ if no matching $M\in \mathcal{M}_k^{\text{SD}}$ has a negative reduced cost, i.e., if for all matchings $M \in \mathcal{M}_k^{\text{SD}}$ 
\begin{equation}
\label{con:violation_dual}
\sum_{(i,j) \in N \times O} -m^t_{ij} (u^*_{ij} + v^*_{ij}) - w^* \geq 0.
\end{equation}

\subsection{Minimizing the deviation from~\texorpdfstring{$X$}{X}: pricing problem}
\label{subsec:pricing_problem}
In order to determine the existence of a matching that violates (\ref{con:violation_dual}), we formulate a pricing problem to find a feasible ex-post efficient matching $M \in \mathcal{M}_k^{\text{SD}}$ that minimizes the left-hand side of inequality~(\ref{con:violation_dual}) for given dual values $u^*$, $v^*$, and $w^*$. This pricing problem contains two sets of constraints: constraints (\ref{con:pp01})-(\ref{con:pp04}) ensure that the constructed matching is feasible and assigns at least $k$ agents to an object, while constraints (\ref{con:pp_Biro1})-(\ref{con:pp_Biro9}) guarantee that the matching is ex-post efficient. 

First, define binary decision variables $m_{ij}$, where $m_{ij} = 1$ if agent~$i \in N$ is assigned to object~$j\in O$, and $m_{ij} = 0$ otherwise. Then, the following integer linear program~$\left[\text{PP\textsubscript{0}}\right]$ will find a minimum-weight matching that assigns at least $k$ agents to an object.
\begin{subequations}
	\begin{align}
	&\left[\text{PP\textsubscript{0}}\right] & &\min & \sum_{i \in N} \sum_{j \in O} -m_{ij} \big(u^*_{ij} &+ v^*_{ij}\Big) - w^*  & \label{con:pp_obj}\\ 
	& & &\: \text{s.t.} &\sum_{j \in O} m_{ij}&\leq 1 &\forall \; i \in N,\label{con:pp01}\\
	& & & &\sum_{i \in N} m_{ij} &\leq q_j &\forall \; j \in O, \label{con:pp02}\\
	& & & &\sum_{i \in N} \sum_{j \in O} m_{ij}&\geq k,\label{con:pp03}\\
	& & & &m_{ij} &\in \{0,1\} &\forall \; (i,j) \in N \times O. \label{con:pp04}
	\end{align}
\end{subequations}
The objective function of $\left[\text{PP}_0\right]$ minimizes the left-hand side of inequality~(\ref{con:violation_dual}), and will find a matching with a negative reduced cost if and only if the objective value is strictly negative. In order to reduce the number of variables in~$\left[\text{PP}_0\right]$ we can fix $m_{ij}$ to one (zero) if agent~$i$ is always (never) assigned to object~$j$ by the probabilistic assignment~$X$ that we want to decompose. Note that $\left[\text{PP\textsubscript{0}}\right]$ is, except for the inequality in constraint~(\ref{con:pp03}), identical to the formulation by \cite{dell1997k} for the $k$-\textsc{cardinality Assignment Problem} ($k$-AP), which finds a minimum-weight matching that assigns exactly $k$ agents to an object. They showed that $k$-AP is polynomially solvable because of the total unimodularity of the constraint matrix. The same holds for $\left[\text{PP\textsubscript{0}}\right]$ if we do not require the final matching to be ex-post efficient.

Additionally, we enforce the matching to be ex-post efficient by appending to $\left[\text{PP\textsubscript{0}}\right]$ the constraints that were recently formulated by \cite{biro2021complexity}, slightly simplified to the setting with strict agent preferences. They showed that a matching $M\in\mathcal{M}$ is ex-post efficient if and only if there exists a vector of prices $p\in\{0,\ldots,|O|\}^{|O|}$ such that $(M,p)$ is a \textit{competitive equilibrium}. We say that $(M,p)$ is a competitive equilibrium if it satisfies the following two conditions:
\begin{enumerate}[label=(\roman*),leftmargin=*,labelindent=0pt]
    \itemsep 0em
    \item each object $j\in O$ with remaining capacity has a price equal to zero: $q_j - \sum_{i\in N}m_{ij} > 0 \Rightarrow p_j = 0$;
    \item each object $k\in O$ that is preferred by at least one agent to its assigned object $j\in O$ has a strictly higher price than~$j$: $\exists i\in N: k >_i j = M(i) \Rightarrow p_k > p_j$.
\end{enumerate}
Note that we do not require the third set of conditions by \cite{biro2021complexity}, because we assume agent preferences to be strict. Alternatively, we have also tried to enforce ex-post efficiency by translating the three necessary and sufficient conditions for ex-post efficiency by \cite{abraham2004pareto} into linear constraints, but the recent formulation by Biró and Gudmundsson outperformed this alternative formulation.

In order to implement the conditions for a competitive equilibrium in a linear model, we introduce some additional decision variables. Using the notation by Bir\'{o} and Gudmundsson, define integer variables $p_j\in\{0,\ldots,|O|\}$ which represent the price for each object $j\in O$. Moreover, define an auxiliary variable $f_j\in\{0,1\}$ for each object $j \in O$, and auxiliary variables $s_{jk} \in\mathbb{N}$ and $\tilde{s}_{jk} \in\{0,1\}$ for each object pair $\{j,k\}\subseteq O$. Denoting the set of objects that agent~$i$ prefers to the outside option by $O_i \subseteq O$, we append the following constraints to $\left[\text{PP}_0\right]$.
\addtocounter{equation}{-1}
\begin{subequations}
\addtocounter{equation}{5}
	\begin{align}
	& \left[\text{PP}_{\text{PE}}\right] & \sum_{i \in N:k>_i j} m_{ij} &= s_{jk}  &\forall \; \{j,k\} \subseteq O, \label{con:pp_Biro1}\\
	& & \tilde{s}_{jk} &\leq s_{jk} &\forall \; \{j,k\} \subseteq O, \label{con:pp_Biro2}\\
	& & \tilde{s}_{jk} \cdot |N|&\geq s_{jk} &\forall \; \{j,k\} \subseteq O, \label{con:pp_Biro3}\\
	& & f_jq_j &\leq \sum_{i \in N} m_{ij} &\forall \; j \in O, \label{con:pp_Biro4}\\
	& & f_j + q_j &\geq \sum_{i \in N} m_{ij} + 1 &\forall \; j \in O, \label{con:pp_Biro5}\\
	& & \sum_{l\in O_i} m_{il} + f_j &\geq 1 &\forall \; (i,j) \in N\times O_i, \label{con:pp_Biro6}\\
	& & f_j \cdot |O| &\geq p_j &\forall \; j \in O, \label{con:pp_Biro7}\\
	& & (1-\tilde{s}_{jk})(|O| + 1) + p_k &\geq p_j + 1 &\forall \; \{j,k\} \subseteq O, \label{con:pp_Biro8}\\
	& & p_j &\in \mathbb{N} &\forall \; j \in O. \label{con:pp_Biro9}\end{align}
\end{subequations}
In summary, constraints (\ref{con:pp_Biro1})-(\ref{con:pp_Biro5}) define the auxiliary variables, constraints (\ref{con:pp_Biro6}) enforce the resulting matching to be maximal, and constraints (\ref{con:pp_Biro7}) and~(\ref{con:pp_Biro8}) ensure that $(M,p)$ is a competitive equilibrium. More specifically, the first set of constraints~(\ref{con:pp_Biro1}) sets~$s_{jk}$ equal to the number of agents that prefer object~$k$ to their assigned object~$j$. Secondly, constraints~(\ref{con:pp_Biro2}) and~(\ref{con:pp_Biro3}) ensure that~$\tilde{s}_{jk}$ equals one if $s_{jk}$ is strictly positive, and that $\tilde{s}_{jk}$ equals zero otherwise. Thirdly, constraints~(\ref{con:pp_Biro4}) and (\ref{con:pp_Biro5}) ensure that~$f_j$ equals zero when object~$j\in O$ has unused capacity, and that $f_j$ equals one otherwise. Constraints~(\ref{con:pp_Biro6}) enforce the matching to be maximal: if agent~$i\in N$ is not assigned to any object, then none of the objects that she prefers to the outside option can have unused capacity. Lastly, constraints~(\ref{con:pp_Biro7}) force the price of object $j\in O$ to be zero when it has unused capacity, while constraints~(\ref{con:pp_Biro8}) impose $p_k \geq p_j + 1$ whenever~$\tilde{s}_{jk} = 1$. Note that adding the integrality conditions~(\ref{con:pp_Biro9}) for~$p_j$ is not required for the correctness of the formulation, but we found that doing so results in lower computation times.

While the basic pricing problem $\left[\text{PP}_0\right]$ can be solved in polynomial time, the addition of constraints (\ref{con:pp_Biro1})-(\ref{con:pp_Biro9}) to $\left[\text{PP}_0\right]$ causes the problem to become computationally hard. This follows directly from the $\mathcal{NP}$-hardness of finding an ex-post efficient matching of maximum weight \citep{biro2021complexity}. 

\subsection{Maximizing the fraction of large matchings}
\label{subsec:alternative_formulation}
In this section, we propose an alternative procedure that obtains an exact decomposition of the probabilistic assignment~$X$ over ex-post efficient matchings, while minimizing the cumulative weight of the matchings in the decomposition that assign less than~$k$ agents to an object. 

Consider the same notation as in Section~\ref{subsec:RMP}, and let $k\in\mathbb{N}$ be the the number of assigned agents for which we are checking whether a given probabilistic assignment~$X\in\Delta\mathcal{M}$ can be decomposed over~$\mathcal{M}^\text{SD}_k$. Denote the cumulative weight of the matchings in the decomposition that assign at least~$k$ agents to an object by decision variable~$\alpha\in[0,1]$. We can slightly adapt $\left[\text{RMP}\right]$ to model the restricted master problem for the alternative formulation as follows.
\begin{subequations}
	\begin{align}
	&\left[\alpha\text{RMP}\right] & &\max &\alpha &&\\ 
	& & &\: \text{s.t.} &\sum_{t:M^t \in \tilde{\mathcal{M}}^{\text{SD}}} \lambda^t m_{ij}^{t} &= x_{ij}  &\forall \; (i,j) \in N \times O,\label{con:alpha_rmp_1}\\
	& & & &\sum_{t:M^t\in\tilde{\mathcal{M}}^{\text{SD}}_k}\lambda^t - \alpha &\geq 0,&\label{con:alpha_rmp_2}\\
	& & & &\sum_{t:M^t \in \tilde{\mathcal{M}}^{\text{SD}}} \lambda^t &= 1,&\label{con:alpha_rmp_3}\\
	& & & &\lambda^{t} &\geq 0  &\forall \; t: M^t \in \tilde{\mathcal{M}}^{\text{SD}}. \label{con:alpha_rmp_4}
	\end{align}
\end{subequations}
Note that constraints (\ref{con:alpha_rmp_1}), (\ref{con:alpha_rmp_3}), and (\ref{con:alpha_rmp_4}) are equivalent to their counterparts in $\left[\text{RMP}\right]$ where~$\tilde{\mathcal{M}}^\text{SD}_k$ is replaced by~$\tilde{\mathcal{M}}^\text{SD}$. Denote the dual variables of constraints~(\ref{con:alpha_rmp_1}), (\ref{con:alpha_rmp_2}), and (\ref{con:alpha_rmp_3}) by~$u_{ij}$, $v\leq-1$, and~$w$ respectively, for each $(i,j)\in N\times O$. Given a solution of $\left[\alpha\text{RMP}\right]$ with dual variables $u^*=[u^*_{ij}]$, $v^*$ and~$w^*$, this solution is optimal if no matching~$M\in\mathcal{M}^\text{SD}$ has a negative reduced cost, i.e., if
\begin{equation}
\label{eq:alternative_pricing_reduced_cost}
\begin{cases}
    \sum\limits_{(i,j)\in N\times O}m_{ij}u^*_{ij} + w < 0 &\forall \; M\in\mathcal{M}^{\text{SD}}\setminus \mathcal{M}^{\text{SD}}_k,\\
    \sum\limits_{(i,j)\in N\times O}m_{ij}u^*_{ij} + v + w< 0 &\forall \; M\in\mathcal{M}^{\text{SD}}_k.
\end{cases}
\end{equation}
Hence, two pricing problems should be solved to identify whether a solution for $\left[\alpha\text{RMP}\right]$ is optimal. The first pricing problem minimizes the first expression in~(\ref{eq:alternative_pricing_reduced_cost}) under constraints (\ref{con:pp01})-(\ref{con:pp_Biro9}), but without cardinality constraint~(\ref{con:pp03}). The second pricing problem simply minimizes the second expression in~(\ref{eq:alternative_pricing_reduced_cost}) under constraints (\ref{con:pp01})-(\ref{con:pp_Biro9}). If any of the pricing problems finds a matching with a negative objective value, it is added to~$\tilde{\mathcal{M}}^{\text{SD}}$ and to~$\tilde{\mathcal{M}}^{\text{SD}}_k$ if the matching assigns at least~$k$ agents, and $\left[\alpha\text{RMP}\right]$ is solved again. If the objective value for both pricing problems is non-negative, the found solution for $\left[\alpha\text{RMP}\right]$ is optimal.

\subsection{Flexibility of the column generation framework}
\label{subsec:flexibility}
    In this subsection, we discuss how the inherent flexibility of linear (integer) programming methods can be exploited to adapt the proposed column generation frameworks to find decompositions of probabilistic assignments in different problem settings. As such, our column generation framework can be readily adapted to alternative assignment settings where ties have to be broken, or where probabilistic solution concepts are preferred. Next to RSD and PS, interesting probabilistic assignments to be decomposed are proposed by \cite{ashlagi2016optimal}, \cite{ashlagi2020assignment}, \cite{kavitha2011popular}, \cite{kesten2017efficient}, and \cite{shi2022optimal}.
    
    Changes in the problem setting can be classified into two categories. First, the decision-maker could require the matchings in the decomposition to satisfy other ex-post properties than ex-post efficiency. Examples of possible alternative properties are quota, stability notions, or enforcing marriage partners to be matched together \citep[e.g.,][]{bronfman2018redesigning}. These alternative properties can be easily enforced by including the relevant constraints in the pricing problem. Note, however, that it might be possible that a given probabilistic assignment is not decomposable over a set of matchings with the desired properties \citep[e.g.,][Theorem 4.1]{bronfman2018redesigning}. In that case, the column generation framework $[\text{RMP}]$ will output a decomposition in which the largest difference between the resulting allocation probabilities and the initial allocation probability is minimized over all agent-object pairs. By changing the objective function of $[\text{RMP}]$, other ways of approximating the initial allocation probabilities can also be modeled.
    
    Second, the decision-maker could be interested in finding a decomposition that maximizes an alternative worst-case measure. Examples of alternative measures include a minimization of the worst-case total travel distance to the schools \citep[e.g.,][]{biro2021complexity}, a minimization of some worst-case \textit{unpopularity} measure \citep[e.g.,][]{kondratev2022minimal, mccutchen2008least}, or an optimization of the worst-case average assigned preference or utility of the agents. These alternative worst-case measures can be easily applied by modifying the binary search procedure. More specifically, we replace constraint~(\ref{con:pp03}) by a set of linear constraints to bound the matchings with respect to the evaluated value of the alternative worst-case measure in the binary search.
        
    While obtaining such a set of linear constraints is rather straightforward for most of the discussed worst-case measures, we illustrate how to obtain a matching with a bounded \textit{unpopularity margin}. Denote by~$\phi(M,M')$ the number of agents who prefer matching~$M$ over matching~$M'$. We say that matching~$M$ is \textit{more popular} than matching~$M'$ if $\phi(M,M')>\phi(M',M)$. The unpopularity margin of a matching~$M$ is then defined as $g(M)=\max_{M'\in\mathcal{M}}\left(\phi(M',M)-\phi(M,M')\right)$. The interpretation of a matching~$M$ having unpopularity margin~$g(M)=\omega$ is that for all matchings~$M'$, the number of agents who prefer~$M'$ over~$M$ is at most~$\omega$ units larger than the number of agents who prefer~$M$ over~$M'$.
    
    To impose the matchings in the decomposition to have an unpopularity margin of at most~$\omega$, we adopt a similar approach as in the recent work by \cite{kavitha2022popular}. Consider the additional binary variable~$m_{i\varnothing}$ which equals one if agent $i\in N$ is assigned to the outside option, and zero otherwise, and let $q_\varnothing = |N|$. Moreover, given a matching~$M\in\mathcal{M}$, let weights $\nu_M(i,j)$ be equal to one if agent~$i$ prefers object~$j$ to $M(i)$, and let $\nu_M(i,j)$ equal minus one if agent~$i$ prefers $M(i)$ to~$j$. Following \cite{kavitha2022popular}, the optimal objective function of formulation $[\text{LP}_1]$, in which $m'_{ij}$ are the decision variables, will be equal to the unpopularity margin of a given matching~$M$.
    \begin{subequations}
	\begin{align}
	&\left[\text{LP}_1\right] & &\max &\sum_{(i,j) \in N \times (O\cup\varnothing)} \nu_M(i,j) \cdot m'_{ij} &&\\ 
	& & &\: \text{s.t.} &\sum_{j\in O\cup\varnothing} m'_{ij} &= 1   &\forall \; i \in N ,\label{con:LP1_1}\\
	& & & &\sum_{i \in N} m'_{ij} &\leq q_j   &\forall \; j \in O\cup\varnothing ,\label{con:LP1_2}\\
	 & & & & m'_{ij} &\geq 0 &\forall \; (i,j) \in N \times (O\cup\varnothing).\label{con:LP1_3}
	\end{align}
\end{subequations}

However, when matching~$M$ is no longer given, but determined upon by the pricing problem, weights $\nu_M$ become decision variables as well. Simply adding a constraint to the pricing problem that upper bounds the objective function of $[\text{LP}_1]$ by~$\omega$ would therefore result in a non-linear integer program. To overcome this issue, we can use strong duality and include the dual constraints of $[\text{LP}_1]$ while lower bounding the dual objective by~$\omega$. As a result, the unpopularity margin of the found matching will be at most~$\omega$ if and only if there exists a vector~$\alpha\in\mathbb{R}^{|N|+|O| + 1}$ and a matching $M=[m_{ij}]\in\mathcal{M}$ that satisfy constraints (\ref{eq:unpop1})-(\ref{eq:unpop4}), where $\alpha$ represents the dual variables of constraints (\ref{con:LP1_1}) and (\ref{con:LP1_2}). Constraint~(\ref{eq:unpop2}) corresponds to the bounded dual objective of $[\text{LP}_1]$, and constraints (\ref{eq:unpop1}) and (\ref{eq:unpop4}) to the dual constraints. Constraints~(\ref{eq:unpop3}) determine the values of weights~$\nu_M$ as defined.
\begin{subequations}
	\begin{align}
	       \sum_{i\in N} \alpha_i + \sum_{j \in O\cup\varnothing}q_j\alpha_j &\geq \omega, & \label{eq:unpop2}\\
	       \alpha_i + \alpha_j &\geq \nu_M(i,j) &\forall \; (i,j) \in N \times \left(O \cup \varnothing\right) \label{eq:unpop1},\\
	       \sum_{k\in O\cup\varnothing:j >_i k}m_{ik}- \sum_{k\in O\cup \varnothing:k>_i j}m_{ik} &= \nu_M(i,j) &\forall \;(i,j) \in N \times (O\cup \varnothing),\label{eq:unpop3}\\
	       \alpha_j&\geq 0 &\forall j \in O\cup\varnothing.\label{eq:unpop4}
	\end{align}
\end{subequations}

\section{Computational experiments}
\label{sec:results}
In this section, we discuss the computational performance of the introduced methods for decomposing $X^{\text{RSD}}$, which we test both on existing and on newly generated instances. The choice to apply our methods on the RSD mechanism is motivated by its widespread use in practical applications. In particular, these experiments allow us to gain the two following main insights. Firstly, despite the worst-case result of Theorem~\ref{theorem:LB-RSD_2}, the optimal value of MD-SD$(X^{\text{RSD}})$ is equal to its upper bound $\lfloor \mu(X^{\text{RSD}}) \rfloor$ for all generated instances that were solved to optimality. Secondly, we observe that even when the column generation frameworks from Section~\ref{sec:column_gen} do not obtain a guaranteed optimal solution within the imposed runtime limit, their results can still be highly valuable in real-world applications, and we discuss two real-world school choice cases from the Belgian cities of Antwerp and Ghent as examples.

\subsection{Implementation details}
\label{subsec:impl_details}
All experiments were implemented with C++, compiled with Microsoft Visual C++ 2019, and run on an Intel Core i7-7700 processor running at 3.60 GHz, with 16GB of RAM memory on a Windows 10 64-bit OS. All linear and integer programs are solved using IBM ILOG CPLEX 12.9, implemented in C++ with Concert Technology, with default parameter settings, and with a precision of~$10^{-4}$ to avoid numerical issues.

In order to decompose $X^{\text{RSD}}$, this matrix first needs to be computed. Although the RSD algorithm can be easily executed by randomly selecting an ordering $\sigma \in \Sigma'$, \cite{aziz2013computational} showed that computing $X^{\text{RSD}}$ is $\#\mathcal{P}$-complete, and thus intractable. Moreover, \cite{saban2015complexity} found that $X^{\text{RSD}}$ cannot even be approximated efficiently. We therefore estimate $X^{\text{RSD}}$ in this section by computing Equation~(\ref{eq:RSD}) over a random sample $\hat{\Sigma}' \subseteq \Sigma'$ of the orderings, and we opt for sample size $|\hat{\Sigma}'| = 10,000$. Because $\hat{\Sigma}'$ is a random sample, both the resulting estimate of $X^{\text{RSD}}$ and the solution of the column generation framework are unbiased.

Moreover, we determine $p^-(I)$, i.e., the cardinality of the minimum-cardinality ex-post efficient matching in a one-sided matching instance $I \in \mathcal{I}$, by formulating an IP model with constraints $[\text{PP}_0]$ and $[\text{PP}_\text{PE}]$ from the pricing problem while minimizing the number of assigned agents. Note that despite the $\mathcal{NP}$-hardness of this problem \citep{abraham2004pareto}, the running times are acceptable for the considered instances.

In all experiments, we include an initial subset of matchings as variables in the initial restricted master problem~$\left[\text{RMP}\right]$ or~$\left[\alpha\text{RMP}\right]$. Here, a trade-off emerges: whereas a larger sample size is likely to decrease the number of times the pricing problem will be called, it also tends to increase the solution time of the restricted master problem. To balance these two effects, we adopt a sample size of $10,000$ random orderings. For column generation framework $\left[\text{RMP}\right]$, described in Sections~\ref{subsec:RMP} and~\ref{subsec:pricing_problem}, we only retain the matchings that assign at least~$k$ agents to an object, where $k\in\mathbb{N}$ is the value to be checked in the binary search. When solving $\left[\text{RMP}\right]$ for $\lfloor\mu(X^{\text{RSD}})\rfloor$, for example, generally slightly more than half of the sampled matchings are retained. Moreover, we start the binary search for both column generation frameworks by setting $k=\lfloor\mu(X^\text{RSD})\rfloor$.

\subsection{Data generation}
\label{subsec:data_gen}
We have developed a parameterized data generation tool that is based on the properties of two real-world data sets from the school choice problem in the Belgian cities Antwerp and Ghent (see Section~\ref{subsec:Antw_Ghent} for more details). Next to basic parameters, such as the objects' capacities or the length of the preference lists, we also consider the popularity of the objects in our data generation tool. More specifically, the user can specify the desired correlation between an object's capacity and its popularity ($\rho$, see Table~\ref{table:datagen_parameters}), to what extent agents who submit a longer preference list include more popular objects on average ($\Delta_1$), and to what extent the popularity of the objects in an agent's preference list is influenced by the popularity of their first choice ($\Delta_2$). Table~\ref{table:datagen_parameters} contains all parameters of the data generation and their default values, which are based on the data of Ghent, and which we used to generate the instances used in the computational experiments. The data generation proceeds as follows.

Denote by $l_i\in\mathbb{N}$ the preference list length of agent~$i \in N$, i.e., the number of objects that are preferred to the outside option by agent~$i$, and denote $l = (l_1, \ldots, l_{|N|})$. Recall that we denote the capacities of the objects by $q = (q_1, \ldots, q_{|O|})$. Moreover, define the \textit{popularity}~$\eta_j\in\mathbb{R}^+$ of an object~$j \in O$ as the ratio of the number of times $j$ is preferred to the outside option over the capacity of that object, and denote $\eta = (\eta_1, \ldots, \eta_{|O|})$. The fraction $\xi\in\left[0,1\right]$ of the objects with the highest popularity are called \textit{popular}. The data generation itself consists of the following four steps.
\begin{enumerate}[label=(\roman*)]
	\itemsep 0em
    \item Generate the preference list lengths~$l_i$ from a truncated normal distribution $\mathcal{N}(\overline{l}, s^2)$ such that $l_i\geq 1$, and round to the closest integer, for all agents~$i \in N$.
    \item Generate $\eta$ and $q$ from truncated standard normal distributions, and ensure that their correlation equals~$\rho$ by multiplying them with the upper triangle matrix in the Cholesky decomposition of the covariance matrix corresponding to $\rho$. 
    \item Rescale $\eta$ and $q$ to satisfy the desired means and standard deviations, and round the elements of~$\eta$ and~$q$ to the closest integers.
    \item Fill in the preference profile~$>$. First, determine for each agent~$i \in N$ and for each position~$t\leq l_i$ in~$>_i$ the probability~$\pi_{it}$ with which a popular object is selected, in function of $\Delta_1$ and $\Delta_2$ (as defined in Table~\ref{table:datagen_parameters}). Second, each object~$j\in O$ in the group of popular, resp.\ unpopular, objects is selected with probability $\pi_{it}\frac{\eta_jq_j}{Q}$, resp.\ $(1-\pi_{it})\frac{\eta_jq_j}{Q}$, where $Q$ is the sum of $\eta_kq_k$ for the (un)popular objects $k\in  O$ that are not yet in $>_i$. If $Q=0$ for the popular (unpopular) objects for some values of $i$ and $t$, then~$\pi_{it}$ is set to zero (one).
\end{enumerate}
\begin{table}[hbt!]
	\small
	\centering
	\caption{Data generation parameters and their default values.}
	\begin{tabular}{c l c}
		\toprule
		\textbf{Parameter} & \textbf{Description} & \textbf{Default value}\\
		\midrule
		$C$ & $\sum_{j\in O} \sfrac{q_j}{|N|}$ & 1.20 \\
		$\overline{l}$ & Mean preference list length~$l$&2.42\\
		$s$ & Standard deviation of preference list length~$l$&1.05\\
		$\xi$ & Fraction of the objects that are popular & 0.10 \\
		$\rho$ & Correlation between $q$ and $\eta$ & 0.21\\
		$CV_c$ &  The coefficient of variation of the capacities~$q$ &0.80\\
		$CV_\eta$ &  The coefficient of variation of the popularity~$\eta$ &0.60\\
		$\Delta_1$ & \multicolumn{1}{p{7.2cm}}{Difference in average popularity of the objects in a preference list of length $\frac{1}{2}(\max_{i \in N}l_i + 1)$, compared to in a preference list of length one}&0.14\\
		$\Delta_2$ &\multicolumn{1}{p{7.2cm}}{Difference in the probability of submitting a popular object between when the first choice is popular and unpopular}&0.01\\
		\bottomrule
	\end{tabular}
	\label{table:datagen_parameters}
\end{table}

A detailed description of the data generation, its repositories and the generated instances are available online (\url{https://github.com/DemeulemeesterT/GOSMI.git}).

\subsection{Observations for generated data}
\label{subsec:observ}

\pgfplotstableset{
	col sep=tab,
	CPU/.style={
		string replace={>1h}{}, 
		string replace={NA}{}, 
		fixed,fixed zerofill,
		column name = {CPU},
		precision=2,
		column type = r,
	},
	NRSUCCES/.style={
		string replace={NA}{},	
		empty cells with={},
		fixed,
		column name = {$\checkmark$},
		column type = r,
	},	
	NROPTISUB_RMP/.style={
		string replace={NA}{},	
		empty cells with={},
		fixed,
		sci,
		precision=1,
		column name = {\multicolumn{1}{c}{$\left[\text{RMP}\right]$}},
		column type = c,
	},	
	NROPTISUB_ALPHA_RMP/.style={
		string replace={NA}{},	
		empty cells with={},
		fixed,
		precision=5,
		column name = {\multicolumn{1}{c}{$\left[\alpha\text{RMP}\right]$}},
		column type = c,
	},
	AGENTCOUNT/.style={
		column name={$|N|$},
		column type = r,
	},
}

\pgfplotstableread[col sep=comma]{ResultsSummary.csv}\Summarytable
	
Table~\ref{table:summary} summarizes the main findings of our computational experiments for generated data. We can conclude that both column generation frameworks that were introduced in Section~\ref{sec:column_gen} manage to solve relatively large instances to optimality in acceptable runtimes. Column generation framework $\left[\text{RMP}\right]$ slightly outperforms $\left[\alpha\text{RMP}\right]$, because it requires less time, on average, and manages to solve more instances to optimality. In general, we observe that both an increase in the number of agents, and an increase in the number of objects per agent cause the runtimes to rise. 

Furthermore, for all instances that were solved to optimality, the difference between $z(X^{\text{RSD}})$ and~$p^-(I)$, the worst-case number of assigned agents by the RSD algorithm, is substantial, with an average increase of 5.42\% of the total number of agents. In fact, this is the maximal possible increase, as we observe that for all optimally solved instances, $z(X^{\text{RSD}})$ is equal to its upper bound $\lfloor \mu(X^{\text{RSD}})\rfloor$, despite the worst-case result of Theorem~\ref{theorem:LB-RSD_2}. This observation can be explained in two different ways. First, \cite{che2010asymptotic} showed that $X^{\text{PS}}$ and $X^{\text{RSD}}$ become equivalent when the market becomes large, and we know from Theorem~\ref{theor:MD-P} that $z(X^{\text{PS}}) = \lfloor \mu(X^{\text{PS}})\rfloor$. Second, an instance similar to the instance~$I_k$ that was given in Theorem~\ref{theorem:LB-RSD_2} to prove the lower bound on $z(X^{\text{RSD}})$ is unlikely to occur in practice, because it requires a fraction $\frac{k^2-k}{k^2}$ of the agents to all submit the same object as their unique preference, and this fraction goes to one for large values of~$k$. 

\begin{table}[tb!]
		\small
		\setlength\tabcolsep{3.7pt}
		\captionsetup{font=footnotesize}
		\caption{Computational results of column generation frameworks $\left[\text{RMP}\right]$ and $\left[\alpha\text{RMP}\right]$ for MD-SD$(X^{\text{RSD}})$. The column CPU displays the average computation time (in s) when optimality was obtained within one hour, and the column~$\checkmark$ contains the number of instances solved to optimality, out of 25 instances.}
		\centering
		\pgfplotstabletypeset[
		columns={N,Time_Orig_1,OPTcount_Orig_1,Time_Alt_1, OPTcount_Alt_1,Time_Orig_10,OPTcount_Orig_10,Time_Alt_10, OPTcount_Alt_10,Time_Orig_50,OPTcount_Orig_50,Time_Alt_50, OPTcount_Alt_50},
		columns/N/.style={AGENTCOUNT},
		columns/Time_Orig_1/.style={CPU},
		columns/Time_Orig_10/.style={CPU},
		columns/Time_Orig_50/.style={CPU},
		columns/Time_Alt_1/.style={CPU},
		columns/Time_Alt_10/.style={CPU},
		columns/Time_Alt_50/.style={CPU},
		columns/OPTcount_Orig_1/.style={NRSUCCES},
		columns/OPTcount_Orig_10/.style={NRSUCCES},
		columns/OPTcount_Orig_50/.style={NRSUCCES},
		columns/OPTcount_Alt_1/.style={NRSUCCES},
		columns/OPTcount_Alt_10/.style={NRSUCCES},
		columns/OPTcount_Alt_50/.style={NRSUCCES},
		every head row/.style={
			before row={%
				\toprule
				& \multicolumn{4}{c}{$\sfrac{|N|}{|O|} = 1$} & \multicolumn{4}{c}{$\sfrac{|N|}{|O|} = 10$} & \multicolumn{4}{c}{$\sfrac{|N|}{|O|} = 50$}
				\\ \cmidrule(rl){2-5} \cmidrule(rl){6-9} \cmidrule(rl){10-13} & \multicolumn{2}{c}{$\;\;\left[\text{RMP}\right]$} & \multicolumn{2}{c}{$\;\left[\alpha\text{RMP}\right]$} & \multicolumn{2}{c}{$\;\;\left[\text{RMP}\right]$} & \multicolumn{2}{c}{$\;\left[\alpha\text{RMP}\right]$} & \multicolumn{2}{c}{$\;\;\left[\text{RMP}\right]$} & \multicolumn{2}{c}{$\;\left[\alpha\text{RMP}\right]$}\\
				\cmidrule(rl){2-3}\cmidrule(rl){4-5}\cmidrule(rl){6-7}\cmidrule(rl){8-9}\cmidrule(rl){10-11}\cmidrule(rl){12-13}
			},
			after row=\midrule,
		},
		every last row/.style={
			after row=\bottomrule,
		},
		every row 5 column 1/.style={postproc cell content/.style={@cell content=\ensuremath{-}}
		},
		every row 5 column 3/.style={postproc cell content/.style={@cell content=\ensuremath{-}}
		},
		every row 5 column 5/.style={postproc cell content/.style={@cell content=\ensuremath{-}}
		},
		every row 5 column 7/.style={postproc cell content/.style={@cell content=\ensuremath{-}}
		},
		]\Summarytable
		\label{table:summary}
	\end{table}

\begin{table}[bt!]
		\small
		\captionsetup{font=footnotesize}
		\caption{Average objective values for the MD-SD$(X^{\text{RSD}})$ instances that could not be solved to optimality within one hour, for column generation frameworks $\left[\text{RMP}\right]$ and $\left[\alpha\text{RMP}\right]$.}
		\centering
		\pgfplotstabletypeset[
		columns={N,Obj_val_unsolved_Orig_1,Obj_val_unsolved_Alt_1,Obj_val_unsolved_Orig_10,Obj_val_unsolved_Alt_10,Obj_val_unsolved_Orig_50,Obj_val_unsolved_Alt_50},
		columns/N/.style={AGENTCOUNT},
		columns/Obj_val_unsolved_Orig_1/.style={NROPTISUB_RMP},
		columns/Obj_val_unsolved_Orig_10/.style={NROPTISUB_RMP},
		columns/Obj_val_unsolved_Orig_50/.style={NROPTISUB_RMP},
		columns/Obj_val_unsolved_Alt_1/.style={NROPTISUB_ALPHA_RMP},
		columns/Obj_val_unsolved_Alt_10/.style={NROPTISUB_ALPHA_RMP},
		columns/Obj_val_unsolved_Alt_50/.style={NROPTISUB_ALPHA_RMP},
		every head row/.style={
			before row={%
				\toprule
				& \multicolumn{2}{c}{$\sfrac{|N|}{|O|} = 1$} & \multicolumn{2}{c}{$\sfrac{|N|}{|O|} = 10$} & \multicolumn{2}{c}{$\sfrac{|N|}{|O|} = 50$}
				\\ \cmidrule(rl){2-3} \cmidrule(rl){4-5} \cmidrule(rl){6-7}
			},
			after row=\midrule,
		},
		skip rows between index={0}{3},
		every last row/.style={
			after row=\bottomrule,
		},
		every row 3 column 1/.style={postproc cell content/.style={@cell content=\ensuremath{-}}
		},
		every row 3 column 3/.style={postproc cell content/.style={@cell content=\ensuremath{-}}
		},
		every row 3 column 5/.style={postproc cell content/.style={@cell content=\ensuremath{-}}
		},
		every row 4 column 5/.style={postproc cell content/.style={@cell content=\ensuremath{-}}
		},
		]\Summarytable
		\label{table:obj_val}
	\end{table}

Table~\ref{table:obj_val} shows that for the instances that could not be solved to optimality within the runtime limit of one hour, the objective values of $\left[\text{RMP}\right]$ and $\left[\alpha\text{RMP}\right]$ are close to their best possible values. Recall that the objective value~$s^*$ of $\left[\text{RMP}\right]$ represents the maximum difference for any agent-object pair between the generated probabilistic assignment by $\left[\text{RMP}\right]$ and $X^{\text{RSD}}$, while the objective value~$\alpha^*$ of $\left[\alpha\text{RMP}\right]$ represents the cumulative weights of the matchings that assign at least $\lfloor\mu(X^{\text{RSD}})\rfloor$ agents to an object in the decomposition of $X^{\text{RSD}}$. For $\left[\text{RMP}\right]$, the worst value of~$s^*$ over all unsolved instances is $2.20\cdot10^{-3}$, for example, whereas the largest value of~$s^*$ over all unsolved instances with $1,000$ agents is only $3.4\cdot10^{-4}$. Similarly, for $\left[\alpha\text{RMP}\right]$, the worst value of $\alpha^*$ over all unsolved instances is 0.92543, whereas the smallest value of~$\alpha^*$ over all unsolved instances with $1,000$ agents is only 0.97272. We argue that, even when optimality cannot be obtained, the resulting decompositions are still valuable for real-world applications. Both methods will return decompositions in which all matchings are ex-post efficient, but the choice of the method depends on the application. If assigning at least $\lfloor\mu(X^{\text{RSD}})\rfloor$ agents to an object is more important than decomposing~$X^{\text{RSD}}$ exactly, then the solution to $\left[\text{RMP}\right]$ provides a decomposition in which the implemented probability for each agent-object pair is at most~$s^*$ higher than in $X^{\text{RSD}}$, because of constraints~(\ref{con:rmp2_1}). Alternatively, if decomposing $X^{\text{RSD}}$ exactly is of higher importance, then $\left[\alpha\text{RMP}\right]$ obtains a decomposition in which a fraction~$\alpha^*$ of the matchings assign at least $\lfloor\mu(X^{\text{RSD}})\rfloor$ agents to an object.

Lastly, we have also evaluated the algorithm we introduced in the proof of Theorem~\ref{theorem:impl_variab} to decompose $X^{\text{RSD}}$ such that all matchings in the decomposition assign at least $z(X^{\text{RSD}})$ agents to an object. For the instances of Table~\ref{table:summary}, the algorithm only resulted in an ex-post efficient decomposition for some instances with a small number of objects, namely for 20 among the 75  instances where $(|N|, |O|) \in \{(10, 10), (50, 5), (100, 2)\}$. The algorithm from Theorem~\ref{theorem:impl_variab} is therefore not suited to find an optimal decomposition for MD-SD$(X)$ when $X$ is not robust ex-post efficient.

\subsection{Antwerp and Ghent}
\label{subsec:Antw_Ghent}
Next to evaluating our methods on generated data, we also consider two real-world school choice data sets from the cities of Antwerp and Ghent. The data from Antwerp corresponds to the primary school enrollment in the scholastic year of 2014-2015, and contains 4,236 students and 186 schools. The data of Ghent corresponds to the secondary school enrollment in 2018-2019, and contains 3,081 students and 64 schools. For confidentiality reasons, the data sets cannot be made publicly available.

In both data sets, the possible increase that we can realize in the worst-case number of assigned agents is substantial, compared to the RSD algorithm. For Antwerp, this increase is equal to 295 students, which is an increase by 6.96\% of the total number of students, while the increase in Ghent equals 162 students, or 5.26\%.

When applying column generation framework $\left[\text{RMP}\right]$ of Sections~\ref{subsec:RMP} and~\ref{subsec:pricing_problem} to both instances, we notice that a larger initial number of included matchings leads to better overall solutions. For Ghent, for example, the best found objective value for $\left[\text{RMP}\right]$ is equal to 0.00161 after twenty-four hours of computing when starting from an initial sample size of 10,000 matchings, among which the matchings that assigned more than $\lfloor\mu(X^{\text{RSD}})\rfloor$ students were retained. When starting from 50,000 matchings, however, an optimal decomposition is found after almost five hours. Similarly, for Antwerp the objective value decreases from 0.00579 after twenty-four hours when sampling 10,000 initial matchings to $8.7\cdot10^{-4}$ after only two hours for an initial sample of 50,000 matchings. 

We observe similar behaviour for the alternative column generation framework $\left[\alpha\text{RMP}\right]$. For Ghent, when initially including $10,000$ matchings, the best found objective value is 0.85873 after twenty-four hours, while an objective value of 0.99900 is found after two hours, in the first iteration, when including $50,000$ matchings. Performing the same test for Antwerp causes the objective value to rise from 0.64919 to 0.93339.

\section{Conclusion}
\label{sec:conclusion}
In this paper, we have studied the problem of decomposing a probabilistic assignment for one-sided matching over ex-post efficient matchings while maximizing the worst-case number of assigned agents. With respect to the two most studied mechanisms for one-sided matching, PS and RSD, we have obtained the following insights. 

We have established that it is always possible to decompose the probabilistic assignment by the PS mechanism in polynomial time over ex-post efficient matchings that all assign the expected number of assigned agents by PS, either rounded up or down. The same result does not hold for the RSD mechanism, however. On the one hand, it is possible for the RSD algorithm to assign only half of the optimal worst-case number of agents to an object. On the other hand, instances exist for which we cannot realize any improvement in the worst-case number of assigned agents when decomposing the probabilistic assignment by RSD over ex-post efficient matchings. 

In our computational experiments, we have found that for all generated instances that were solved to optimality, the optimal worst-case number of assigned agents when decomposing the probabilistic assignment by RSD over ex-post efficient matchings is equal to the expected number of assigned agents by RSD, rounded down. This promising result encourages the adoption of our solution methods in practical applications. By applying the column generation frameworks that we have introduced to real-world school choice data sets, we have illustrated that even if optimality cannot be obtained in the desired runtime, the found solution can still be highly valuable in practice.

\setlength{\parskip}{10pt}
\footnotesize \textbf{Acknowledgements} Tom Demeulemeester is funded by PhD fellowship 11J8721N of Research Foundation - Flanders. Ben Hermans is funded by post-doc fellowship 12ZZI21N of Research Foundation - Flanders. We are grateful to Steven Penneman from the city of Antwerp, and Pieter De Wilde from the city of Ghent for their collaboration and for making their data sets available to us. Moreover, we would like to thank Markus Brill, Ágnes Cseh, and the members of the ALGO research group at TU Berlin for their feedback. Lastly, we would like to thank the anonymous reviewers for their valuable comments and suggestions.
\setlength{\parskip}{0pt}

\normalsize
\bibliography{mybibfile}

\begin{thebibliography}{36}
\expandafter\ifx\csname natexlab\endcsname\relax\def\natexlab#1{#1}\fi
\providecommand{\url}[1]{\texttt{#1}}
\providecommand{\href}[2]{#2}
\providecommand{\path}[1]{#1}
\providecommand{\DOIprefix}{doi:}
\providecommand{\ArXivprefix}{arXiv:}
\providecommand{\URLprefix}{URL: }
\providecommand{\Pubmedprefix}{pmid:}
\providecommand{\doi}[1]{\href{http://dx.doi.org/#1}{\path{#1}}}
\providecommand{\Pubmed}[1]{\href{pmid:#1}{\path{#1}}}
\providecommand{\bibinfo}[2]{#2}
\ifx\xfnm\relax \def\xfnm[#1]{\unskip,\space#1}\fi
\bibitem[{Abdulkadiroğlu \& Sönmez(1998)}]{abdul1998}
\bibinfo{author}{Abdulkadiroğlu, A.}, \& \bibinfo{author}{Sönmez, T.}
  (\bibinfo{year}{1998}).
\newblock \bibinfo{title}{Random {S}erial {D}ictatorship and the core from
  random endowments in {H}ouse {A}llocation problems}.
\newblock {\it \bibinfo{journal}{Econometrica}\/},  {\it
  \bibinfo{volume}{66}\/}, \bibinfo{pages}{689--701}.
\bibitem[{Abraham et~al.(2004)Abraham, Cechl{\'a}rov{\'a}, Manlove \&
  Mehlhorn}]{abraham2004pareto}
\bibinfo{author}{Abraham, D.~J.}, \bibinfo{author}{Cechl{\'a}rov{\'a}, K.},
  \bibinfo{author}{Manlove, D.~F.}, \& \bibinfo{author}{Mehlhorn, K.}
  (\bibinfo{year}{2004}).
\newblock \bibinfo{title}{Pareto optimality in house allocation problems}.
\newblock In {\it \bibinfo{booktitle}{International Symposium on Algorithms and
  Computation}\/} (pp. \bibinfo{pages}{3--15}).
\newblock \bibinfo{organization}{Springer}.
\bibitem[{Akbarpour \& Nikzad(2020)}]{akbarpour2020approximate}
\bibinfo{author}{Akbarpour, M.}, \& \bibinfo{author}{Nikzad, A.}
  (\bibinfo{year}{2020}).
\newblock \bibinfo{title}{Approximate random allocation mechanisms}.
\newblock {\it \bibinfo{journal}{The Review of Economic Studies}\/},  {\it
  \bibinfo{volume}{87}\/}, \bibinfo{pages}{2473--2510}.
\bibitem[{Ashlagi et~al.(2020)Ashlagi, Saberi \&
  Shameli}]{ashlagi2020assignment}
\bibinfo{author}{Ashlagi, I.}, \bibinfo{author}{Saberi, A.}, \&
  \bibinfo{author}{Shameli, A.} (\bibinfo{year}{2020}).
\newblock \bibinfo{title}{Assignment mechanisms under distributional
  constraints}.
\newblock {\it \bibinfo{journal}{Operations Research}\/},  {\it
  \bibinfo{volume}{68}\/}, \bibinfo{pages}{467--479}.
\bibitem[{Ashlagi \& Shi(2016)}]{ashlagi2016optimal}
\bibinfo{author}{Ashlagi, I.}, \& \bibinfo{author}{Shi, P.}
  (\bibinfo{year}{2016}).
\newblock \bibinfo{title}{Optimal allocation without money: An engineering
  approach}.
\newblock {\it \bibinfo{journal}{Management Science}\/},  {\it
  \bibinfo{volume}{62}\/}, \bibinfo{pages}{1078--1097}.
\bibitem[{Athanassoglou \& Sethuraman(2011)}]{athanassoglou2011house}
\bibinfo{author}{Athanassoglou, S.}, \& \bibinfo{author}{Sethuraman, J.}
  (\bibinfo{year}{2011}).
\newblock \bibinfo{title}{House allocation with fractional endowments}.
\newblock {\it \bibinfo{journal}{International Journal of Game Theory}\/},
  {\it \bibinfo{volume}{40}\/}, \bibinfo{pages}{481--513}.
\bibitem[{Aziz(2020)}]{aziz2020simultaneously}
\bibinfo{author}{Aziz, H.} (\bibinfo{year}{2020}).
\newblock \bibinfo{title}{Simultaneously achieving ex-ante and ex-post
  fairness}.
\newblock In {\it \bibinfo{booktitle}{International Conference on Web and
  Internet Economics}\/} (pp. \bibinfo{pages}{341--355}).
\newblock \bibinfo{organization}{Springer}.
\bibitem[{Aziz et~al.(2013)Aziz, Brandt \& Brill}]{aziz2013computational}
\bibinfo{author}{Aziz, H.}, \bibinfo{author}{Brandt, F.}, \&
  \bibinfo{author}{Brill, M.} (\bibinfo{year}{2013}).
\newblock \bibinfo{title}{The computational complexity of {R}andom {S}erial
  {D}ictatorship}.
\newblock {\it \bibinfo{journal}{Economics Letters}\/},  {\it
  \bibinfo{volume}{121}\/}, \bibinfo{pages}{341--345}.
\bibitem[{Aziz et~al.(2014)Aziz, Mackenzie, Xia \& Ye}]{Aziz2014structurearXiv}
\bibinfo{author}{Aziz, H.}, \bibinfo{author}{Mackenzie, S.},
  \bibinfo{author}{Xia, L.}, \& \bibinfo{author}{Ye, C.}
  (\bibinfo{year}{2014}).
\newblock \bibinfo{title}{Structure and complexity of ex post efficient random
  assignments}.
\newblock {\it \bibinfo{journal}{arXiv:1409.6076 [cs.GT]}\/}, .
\bibitem[{Aziz et~al.(2015)Aziz, Mackenzie, Xia \& Ye}]{aziz2015ex}
\bibinfo{author}{Aziz, H.}, \bibinfo{author}{Mackenzie, S.},
  \bibinfo{author}{Xia, L.}, \& \bibinfo{author}{Ye, C.}
  (\bibinfo{year}{2015}).
\newblock \bibinfo{title}{Ex post efficiency of random assignments}.
\newblock In {\it \bibinfo{booktitle}{Proceedings of the 2015 International
  Conference on Autonomous Agents and Multiagent Systems}\/} (pp.
  \bibinfo{pages}{1639--1640}).
\bibitem[{Bertsimas \& Tsitsiklis(1997)}]{bertsimas1997introduction}
\bibinfo{author}{Bertsimas, D.}, \& \bibinfo{author}{Tsitsiklis, J.~N.}
  (\bibinfo{year}{1997}).
\newblock {\it \bibinfo{title}{Introduction to {L}inear {O}ptimization}\/}
  volume~\bibinfo{volume}{6}.
\newblock \bibinfo{publisher}{Athena Scientific Belmont, MA}.
\bibitem[{Bhalgat et~al.(2011)Bhalgat, Chakrabarty \&
  Khanna}]{bhalgat2011social}
\bibinfo{author}{Bhalgat, A.}, \bibinfo{author}{Chakrabarty, D.}, \&
  \bibinfo{author}{Khanna, S.} (\bibinfo{year}{2011}).
\newblock \bibinfo{title}{Social welfare in one-sided matching markets without
  money}.
\newblock In {\it \bibinfo{booktitle}{Approximation, randomization, and
  combinatorial optimization. Algorithms and techniques}\/} (pp.
  \bibinfo{pages}{87--98}).
\newblock \bibinfo{publisher}{Springer}.
\bibitem[{Birkhoff(1946)}]{birkhoff1946}
\bibinfo{author}{Birkhoff, G.} (\bibinfo{year}{1946}).
\newblock \bibinfo{title}{Three observations on linear algebra}.
\newblock {\it \bibinfo{journal}{Universidad Nacional de Tucuman Revista, Serie
  A}\/},  {\it \bibinfo{volume}{5}\/}, \bibinfo{pages}{147--151}.
\bibitem[{Bir{\'o}(2017)}]{biro2017applications}
\bibinfo{author}{Bir{\'o}, P.} (\bibinfo{year}{2017}).
\newblock \bibinfo{title}{Applications of matching models under preferences}.
\newblock In \bibinfo{editor}{U.~Endriss} (Ed.), {\it
  \bibinfo{booktitle}{{Trends in Computational Social Choice}}\/}
  chapter~\bibinfo{chapter}{18}. (pp. \bibinfo{pages}{345--373}).
\newblock \bibinfo{publisher}{AI Access}.
\bibitem[{Bir{\'o} \& Gudmundsson(2021)}]{biro2021complexity}
\bibinfo{author}{Bir{\'o}, P.}, \& \bibinfo{author}{Gudmundsson, J.}
  (\bibinfo{year}{2021}).
\newblock \bibinfo{title}{Complexity of finding {P}areto-efficient allocations
  of highest welfare}.
\newblock {\it \bibinfo{journal}{European Journal of Operational Research}\/},
  {\it \bibinfo{volume}{291}\/}, \bibinfo{pages}{614--628}.
\bibitem[{Bogomolnaia \& Moulin(2001)}]{bogomolnaiaMoulin2001}
\bibinfo{author}{Bogomolnaia, A.}, \& \bibinfo{author}{Moulin, H.}
  (\bibinfo{year}{2001}).
\newblock \bibinfo{title}{A new solution to the random assignment problem}.
\newblock {\it \bibinfo{journal}{Journal of Economic Theory}\/},  {\it
  \bibinfo{volume}{100}\/}, \bibinfo{pages}{295--328}.
\bibitem[{Bogomolnaia \& Moulin(2015)}]{bogomolnaia2015size}
\bibinfo{author}{Bogomolnaia, A.}, \& \bibinfo{author}{Moulin, H.}
  (\bibinfo{year}{2015}).
\newblock \bibinfo{title}{Size versus fairness in the assignment problem}.
\newblock {\it \bibinfo{journal}{Games and Economic Behavior}\/},  {\it
  \bibinfo{volume}{90}\/}, \bibinfo{pages}{119--127}.
\bibitem[{Bronfman et~al.(2018)Bronfman, Alon, Hassidim \&
  Romm}]{bronfman2018redesigning}
\bibinfo{author}{Bronfman, S.}, \bibinfo{author}{Alon, N.},
  \bibinfo{author}{Hassidim, A.}, \& \bibinfo{author}{Romm, A.}
  (\bibinfo{year}{2018}).
\newblock \bibinfo{title}{Redesigning the {I}sraeli medical internship match}.
\newblock {\it \bibinfo{journal}{ACM Transactions on Economics and
  Computation}\/},  {\it \bibinfo{volume}{6}\/}, \bibinfo{pages}{1--18}.
\bibitem[{Budish et~al.(2013)Budish, Che, Kojima \& Milgrom}]{budish2013}
\bibinfo{author}{Budish, E.}, \bibinfo{author}{Che, Y.-K.},
  \bibinfo{author}{Kojima, F.}, \& \bibinfo{author}{Milgrom, P.}
  (\bibinfo{year}{2013}).
\newblock \bibinfo{title}{Designing random allocation mechanisms: Theory and
  applications}.
\newblock {\it \bibinfo{journal}{American Economic Review}\/},  {\it
  \bibinfo{volume}{103}\/}, \bibinfo{pages}{585--623}.
\bibitem[{Che \& Kojima(2010)}]{che2010asymptotic}
\bibinfo{author}{Che, Y.-K.}, \& \bibinfo{author}{Kojima, F.}
  (\bibinfo{year}{2010}).
\newblock \bibinfo{title}{Asymptotic equivalence of probabilistic serial and
  random priority mechanisms}.
\newblock {\it \bibinfo{journal}{Econometrica}\/},  {\it
  \bibinfo{volume}{78}\/}, \bibinfo{pages}{1625--1672}.
\bibitem[{Dell'Amico \& Martello(1997)}]{dell1997k}
\bibinfo{author}{Dell'Amico, M.}, \& \bibinfo{author}{Martello, S.}
  (\bibinfo{year}{1997}).
\newblock \bibinfo{title}{The $k$-cardinality assignment problem}.
\newblock {\it \bibinfo{journal}{Discrete Applied Mathematics}\/},  {\it
  \bibinfo{volume}{76}\/}, \bibinfo{pages}{103--121}.
\bibitem[{Elster(1992)}]{elster1992local}
\bibinfo{author}{Elster, J.} (\bibinfo{year}{1992}).
\newblock {\it \bibinfo{title}{{Local Justice: How Institutions Allocate Scarce
  Goods and Necessary Burdens}}\/}.
\newblock \bibinfo{publisher}{Russell Sage Foundation}.
\bibitem[{Freeman et~al.(2020)Freeman, Shah \& Vaish}]{freeman2020best}
\bibinfo{author}{Freeman, R.}, \bibinfo{author}{Shah, N.}, \&
  \bibinfo{author}{Vaish, R.} (\bibinfo{year}{2020}).
\newblock \bibinfo{title}{Best of both worlds: Ex-ante and ex-post fairness in
  resource allocation}.
\newblock In {\it \bibinfo{booktitle}{Proceedings of the 21st ACM Conference on
  Economics and Computation}\/} (pp. \bibinfo{pages}{21--22}).
\bibitem[{Kavitha et~al.(2022)Kavitha, Kir{\'a}ly, Matuschke, Schlotter \&
  Schmidt-Kraepelin}]{kavitha2022popular}
\bibinfo{author}{Kavitha, T.}, \bibinfo{author}{Kir{\'a}ly, T.},
  \bibinfo{author}{Matuschke, J.}, \bibinfo{author}{Schlotter, I.}, \&
  \bibinfo{author}{Schmidt-Kraepelin, U.} (\bibinfo{year}{2022}).
\newblock \bibinfo{title}{The popular assignment problem: when cardinality is
  more important than popularity}.
\newblock In {\it \bibinfo{booktitle}{Proceedings of the 2022 Annual ACM-SIAM
  Symposium on Discrete Algorithms (SODA)}\/} (pp. \bibinfo{pages}{103--123}).
\newblock \bibinfo{organization}{SIAM}.
\bibitem[{Kavitha et~al.(2011)Kavitha, Mestre \& Nasre}]{kavitha2011popular}
\bibinfo{author}{Kavitha, T.}, \bibinfo{author}{Mestre, J.}, \&
  \bibinfo{author}{Nasre, M.} (\bibinfo{year}{2011}).
\newblock \bibinfo{title}{Popular mixed matchings}.
\newblock {\it \bibinfo{journal}{Theoretical Computer Science}\/},  {\it
  \bibinfo{volume}{412}\/}, \bibinfo{pages}{2679--2690}.
\bibitem[{Kesten et~al.(2017)Kesten, Kurino \& Nesterov}]{kesten2017efficient}
\bibinfo{author}{Kesten, O.}, \bibinfo{author}{Kurino, M.}, \&
  \bibinfo{author}{Nesterov, A.~S.} (\bibinfo{year}{2017}).
\newblock \bibinfo{title}{Efficient lottery design}.
\newblock {\it \bibinfo{journal}{Social Choice and Welfare}\/},  {\it
  \bibinfo{volume}{48}\/}, \bibinfo{pages}{31--57}.
\bibitem[{Kondratev \& Nesterov(2022)}]{kondratev2022minimal}
\bibinfo{author}{Kondratev, A.~Y.}, \& \bibinfo{author}{Nesterov, A.~S.}
  (\bibinfo{year}{2022}).
\newblock \bibinfo{title}{Minimal envy and popular matchings}.
\newblock {\it \bibinfo{journal}{European Journal of Operational Research}\/},
  {\it \bibinfo{volume}{296}\/}, \bibinfo{pages}{776--787}.
\bibitem[{Krysta et~al.(2019)Krysta, Manlove, Rastegari \&
  Zhang}]{krysta2019size}
\bibinfo{author}{Krysta, P.}, \bibinfo{author}{Manlove, D.},
  \bibinfo{author}{Rastegari, B.}, \& \bibinfo{author}{Zhang, J.}
  (\bibinfo{year}{2019}).
\newblock \bibinfo{title}{Size versus truthfulness in the house allocation
  problem}.
\newblock {\it \bibinfo{journal}{Algorithmica}\/},  {\it
  \bibinfo{volume}{81}\/}, \bibinfo{pages}{3422--3463}.
\bibitem[{Martini(2016)}]{martini2016strategy}
\bibinfo{author}{Martini, G.} (\bibinfo{year}{2016}).
\newblock \bibinfo{title}{Strategy-proof and fair assignment is wasteful}.
\newblock {\it \bibinfo{journal}{Games and Economic Behavior}\/},  {\it
  \bibinfo{volume}{98}\/}, \bibinfo{pages}{172--179}.
\bibitem[{McCutchen(2008)}]{mccutchen2008least}
\bibinfo{author}{McCutchen, R.~M.} (\bibinfo{year}{2008}).
\newblock \bibinfo{title}{The least-unpopularity-factor and
  least-unpopularity-margin criteria for matching problems with one-sided
  preferences}.
\newblock In {\it \bibinfo{booktitle}{Latin American Symposium on Theoretical
  Informatics}\/} (pp. \bibinfo{pages}{593--604}).
\newblock \bibinfo{organization}{Springer}.
\bibitem[{Nesterov(2017)}]{nesterov2017fairness}
\bibinfo{author}{Nesterov, A.~S.} (\bibinfo{year}{2017}).
\newblock \bibinfo{title}{Fairness and efficiency in strategy-proof object
  allocation mechanisms}.
\newblock {\it \bibinfo{journal}{Journal of Economic Theory}\/},  {\it
  \bibinfo{volume}{170}\/}, \bibinfo{pages}{145--168}.
\bibitem[{Ramezanian \& Feizi(2022)}]{ramezanian2022robust}
\bibinfo{author}{Ramezanian, R.}, \& \bibinfo{author}{Feizi, M.}
  (\bibinfo{year}{2022}).
\newblock \bibinfo{title}{Robust ex-post pareto efficiency and fairness in
  random assignments: Two impossibility results}.
\newblock {\it \bibinfo{journal}{Available at SSRN 4037462}\/}, .
\bibitem[{Saban \& Sethuraman(2015)}]{saban2015complexity}
\bibinfo{author}{Saban, D.}, \& \bibinfo{author}{Sethuraman, J.}
  (\bibinfo{year}{2015}).
\newblock \bibinfo{title}{The complexity of computing the {R}andom {P}riority
  allocation matrix}.
\newblock {\it \bibinfo{journal}{Mathematics of Operations Research}\/},  {\it
  \bibinfo{volume}{40}\/}, \bibinfo{pages}{1005--1014}.
\bibitem[{Shi(2022)}]{shi2022optimal}
\bibinfo{author}{Shi, P.} (\bibinfo{year}{2022}).
\newblock \bibinfo{title}{Optimal priority-based allocation mechanisms}.
\newblock {\it \bibinfo{journal}{Management Science}\/},  {\it
  \bibinfo{volume}{68}\/}, \bibinfo{pages}{171--188}.
\bibitem[{{von Neumann}(1953)}]{vonNeumann1953}
\bibinfo{author}{{von Neumann}, J.} (\bibinfo{year}{1953}).
\newblock \bibinfo{title}{A certain zero-sum two-person game equivalent to the
  optimal assignment problem}.
\newblock {\it \bibinfo{journal}{Contributions to the Theory of Games}\/},
  {\it \bibinfo{volume}{2}\/}.
\newblock \bibinfo{note}{Edited by W. Kuhn and A.W. Tucker. Princeton:
  Princeton University Press, 1997}.
\bibitem[{Zhou(1990)}]{zhou1990conjecture}
\bibinfo{author}{Zhou, L.} (\bibinfo{year}{1990}).
\newblock \bibinfo{title}{On a conjecture by {G}ale about one-sided matching
  problems}.
\newblock {\it \bibinfo{journal}{Journal of Economic Theory}\/},  {\it
  \bibinfo{volume}{52}\/}, \bibinfo{pages}{123--135}.

\end{thebibliography}

\appendix
\section{Modified algorithm of \texorpdfstring{\cite{budish2013}}{Budish et al.\ (2013)}}
\label{app:budish_modified}
	This appendix describes a modified version of the algorithm of \cite{budish2013} that, given a probabilistic assignment~$X \in \Delta \mathcal{M}$ with~$\mu(X) \in \mathbb{N}$, constructs a matching~$M \in \mathcal{M}$ that satisfies Properties~\ref{prop_cardinality}-\ref{prop_integrality} as stated in the proof of Theorem~\ref{theorem:impl_variab}. Our approach simplifies the one of \cite{budish2013} as it considers only the specific constraint structure~$\mathcal{H}$ with quotas~$(q_S)_{S \in \mathcal{H}}$ corresponding to~$\Delta \mathcal{M}$, as introduced in the proof of Theorem~\ref{theorem:impl_variab}, rather than the more general bihierarchical structure considered by \cite{budish2013}. This restriction enables us to reduce the $\mathcal{O}(\vert\mathcal{H}\vert^2)$ running time of the algorithm by \cite{budish2013} to  $\mathcal{O}(\vert\mathcal{H}\vert\min \{\vert N \vert, \vert O \vert \})$.
	
	Given an assignment~$X \in \Delta \mathcal{M}$, consider a bipartite graph~$B(X) = (V, E(X))$ with vertices $V = N \cup O \cup \{s,t\}$ and a set of edges~$E(X)$ consisting of
	\begin{enumerate}[label = (\roman*)]
		\itemsep 0em
		\item an edge~$\{i,j\}$ for every~$(i,j) \in N \times O$ for which $x_{ij}$ is fractional,
		\item an edge~$\{s, i\}$ for every~$i \in N$ for which $\sum_{j \in O} x_{ij}$ is fractional,
		\item and an edge~$\{j, t\}$ for every~$j \in O$ for which $\sum_{i \in N} x_{ij}$ is fractional.
	\end{enumerate}

    Every iteration~$u$ of the algorithm starts from a given probabilistic assignment~$X^u$  and a graph~$B(X^u)$, where~$X^u = X$ if $u = 1$, and $X^u$ is defined in the previous iteration otherwise. If~$X^u$ is a matching, i.e., if $X^u \in \mathcal{M}$, then the procedure outputs~$M = X^u$ and terminates. Otherwise, we determine a cycle $C = \langle i_1, j_1, i_2, j_2, \ldots, i_{r}, j_{r}\rangle$ in~$B(X^u)$ with $\{i_k, j_k\}, \{j_k, i_{k + 1}\} \in E(X^u)$ for every $k = 1, \ldots, r$ and $i_{r + 1} = i_{1}$, where the indices are such that $i_k \in N \cup \{t\}$ and~$j_k \in O \cup \{s\}$ for every~$k = 1,\ldots, r$. Next, we let~$\alpha$ be the largest real number such that the probabilistic assignment~$X^{u + 1}$ with
    \begin{equation*}
    	x^{u + 1}_{ij} = 
    	\begin{cases} 
    		x^u_{ij} + \alpha & \text{if $(i,j) = (i_k, j_k)$ for some $k = 1, \ldots, r$,} \\
    		x^u_{ij} - \alpha & \text{if $(i,j) = (i_{k + 1}, j_{k})$ for some $k = 1, \ldots, r$,} \\
    		x^{u}_{ij} & \text{otherwise,}
    	\end{cases}
    \end{equation*}
    is feasible, update the graph~$B(X^{u + 1})$, and proceed to the next iteration.

    To see that the algorithm yields a matching with the desired properties in time $\mathcal{O}(\vert \mathcal{H} \vert \linebreak[1] \min \{\vert N \vert, \vert O \vert\})$, observe first that in every iteration~$u$ the constructed assignment $X^u \in \Delta \mathcal{M}$ satisfies Properties~\ref{prop_cardinality}-\ref{prop_integrality} as stated in the proof of Theorem~\ref{theorem:impl_variab}. Moreover, it follows from the construction of~$B(X^u)$ that the degree of each vertex in~$B(X^u)$ is either zero or at least two and, since~$B(X^u)$ is bipartite, that each cycle contains at most~$2 \min\{\vert N \cup \{t\}\vert, \vert O \cup \{s\}\vert\}$ edges. Hence, if there is at least one~$(i,j) \in N \times O$ for which~$x^{u}_{ij}$ is fractional, then we can obtain a cycle of the desired form, compute~$\alpha$, and determine~$X^{u + 1}$ and~$B(X^{u+1})$ in time $\mathcal{O}(\min \{\vert N \vert, \vert O \vert\})$. Finally, since~$\alpha$ is chosen such that~$\tau(X^{u + 1}) > \tau(X^u)$, where~$\tau$ is defined as in the proof of Theorem~\ref{theorem:impl_variab}, the algorithm terminates after at most~$|\mathcal{H}|$ iterations.
    
    \section{Proof of Theorem~\ref{theorem:LB-RSD_2}}\label{appendix:proofThLB}
    Consider a family of one-sided matching instances~$I_k = (N, O, >, q)$, with $k \geq 2$ an integer. Let $N = \{1, 2, \ldots, k^2\}$ be a set of $k^2$ agents, and let $O = \{o_1, o_2, \ldots, o_{k+1}\}$ be a set of $k+1$ objects. Let the capacities of the objects in $O$ be equal to $q = (k, 1, 1, \ldots, 1)$. Moreover, let the preference list $>_i$ of each agent $i \in N$ be equal to
	\begin{equation*}
	>_i \; = 
	\begin{cases}
	o_1 > o_2 > \ldots > o_{k+1} > \varnothing &\text{if } i \leq k,\\
	o_1 > \varnothing > o_2 > \ldots > o_{k+1} &\text{if } i > k.
	\end{cases}
	\end{equation*}
	
	We now show that $\mu(X^{\text{RSD}(I_k)}) =  2k-1$ and $z(X^{\text{RSD}(I_k)}) = k$ for every~$k \geq 2$. First, consider an arbitrary strict ordering~$\sigma \in \Sigma'$ of the agents that is used by the SD mechanism to construct an ex-post efficient matching $\text{SD}(I_k, \sigma) \in \mathcal{M}^{\text{SD}}$. The preference lists are such that each agent~$i \leq k$ always receives an object and such that an agent~$i > k$ only receives an object if it is among the first~$k$ agents in the ordering~$\sigma$. 	Hence, if we denote by~$\tau_i$ the probability that agent~$i$ receives an object by applying the SD mechanism, then we obtain that~$\tau_i = 1$ if $i \leq k$, and, since there are~$k^2$ agents, that~$\tau_i = k / k^2 = 1/k$ if $i > k$. It follows that
	\begin{equation*}
		\mu(X^{\text{RSD}(I_k)}) = \sum_{i = 1}^{k^2} \tau_i = \sum_{i=1}^k \tau_i + \sum_{i=k+1}^{k^2} \tau_i = k + \frac{k^2 - k}{k} = 2k -1.
	\end{equation*}

	Second, to prove that $z(X^{\text{RSD}(I_k)}) = k$, recall that the set of all ex-post efficient matchings that assign at least $r \in \mathbb{N}$ agents to an object is denoted by $\mathcal{M}^{\text{SD}}_r \subseteq \mathcal{M}^{\text{SD}}$. In instance~$I_k$, denote the unique ex-post efficient matching that assigns $k$ agents to an object by $M^0 \in \mathcal{M}^{\text{SD}}$, and note that all other matchings in $\mathcal{M}^{\text{SD}}$ assign strictly more than $k$ agents to an object. Moreover, note that matching SD$(I_k, \sigma)$ will only be equal to $M^0$ if all agents $i\leq k$ are ranked before all $k^2-k$ other agents $i>k$ in ordering~$\sigma$. Denote the probability that the RSD algorithm will obtain matching~$M^0$ in instance~$I_k$ by $\theta_k = \binom{k^2}{k}^{-1} \!> 0$. Denoting the weight of matching $M^t\in \mathcal{M}^{\text{SD}}$ by $\lambda^t\geq 0$, every feasible decomposition of $X^{\text{RSD}(I_k)}$ should satisfy
	\begin{equation}
		\label{eq:Claim2_proof}
		\sum_{t:M^t \in \mathcal{M}^{\text{SD}}_{k+1}} \lambda^t = \sum_{i = 1}^k x^{\text{RSD}(I_k)}_{i2} = 1 - \theta_k < 1,
	\end{equation}
	because we know that every matching in $\mathcal{M}^{\text{SD}}_{k+1}$ assigns exactly one agent $i \leq k$ to object~$o_2$. As $\mathcal{M}^{\text{SD}} = \mathcal{M}_{k+1}^{\text{SD}} \cup M^0$ and $M^0 \notin \mathcal{M}_{k+1}^{\text{SD}}$, Equation~(\ref{eq:Claim2_proof}) implies that $\lambda^0 > 0$. This, in turn, implies that $z(X^{\text{RSD}(I_k)}) = k$, because every feasible decomposition of $X^{\text{RSD}(I_k)}$ over ex-post efficient matchings has a strictly positive weight for the matching~$M^0$ that assigns exactly $k$ agents to an object.
		
	Combining equalities $\mu(X^{\text{RSD}(I_k)}) =  2k-1$ and $z(X^{\text{RSD}(I_k)}) = k$, we obtain that
	\begin{equation*}
	\lim_{k \to \infty}\frac{z(X^{\text{RSD}(I_k)})}{\lfloor \mu(X^{\text{RSD}(I_k)}) \rfloor}  = \lim_{k \to \infty}\frac{k}{2k - 1} = \frac{1}{2},
	\end{equation*}	
	which proves the theorem.
	
     \section{Proof of Theorem~\ref{theorem:UB-RSD_2}}\label{appendix:proofThUB}
    Consider a family of one-sided matching instances~$I_\ell = (N,O,>, q)$, with $\ell \geq 2$ an integer. Let $N = \{1, 2, \ldots, \ell^2\}$ be a set of $\ell^2$ agents, and let $O = \{o_1, o_2\}$ be a set of two objects. Let the capacities of the objects in $O$ be equal to $q = (\ell, \ell)$. Moreover, let the preference list $>_i$ of each agent $i \in N$ be equal to
	\begin{equation*}
	>_i \; = 
	\begin{cases}
	o_1 > o_2 > \varnothing &\text{if } i \leq \ell,\\
	o_1 > \varnothing > o_2 &\text{if } i > \ell.
	\end{cases}
	\end{equation*}
	
	First, the minimum-cardinality ex-post efficient matching in this instance will be equal to SD$(I_\ell, \sigma)$ for every strict ordering $\sigma \in \Sigma'$ that ranks all agents $i \leq \ell$ before agents $i > \ell$. Hence, $p^-(I_\ell) = \ell$. 
	
	Second, we show that we can decompose $X^{\text{RSD}(I_\ell)}$ over ex-post efficient matchings that all assign $2\ell-1$ agents to an object, i.e., $z(X^{\text{RSD}(I_\ell)}) = 2\ell - 1$. To construct a decomposition of $X^{\text{RSD}(I_\ell)}$ over $\mathcal{M}^{\text{SD}}_{2\ell-1}$, consider the following set of $\ell$ matchings $\{M^t\}_{t=1}^\ell$, where
	    \begin{equation*}
		M^t(i) = 
		\begin{cases}
		o_1 &\text{if } i = t,\\
		o_2 &\text{if } i \leq \ell \text{ and } i \neq t,\\
		o_1 &\text{if } i \in \cup_{\gamma = 2}^{\ell}	\{(\ell-1)t + \gamma\},\\
		\varnothing &\text{otherwise}.		
		\end{cases}
		\end{equation*}
	Clearly, all matchings in $\{M^t\}_{t=1}^\ell$ are ex-post efficient and assign exactly $2\ell-1$ agents to an object. Note that each agent is assigned to $o_1$ by exactly one matching in $\{M^t\}_{t=1}^\ell$, and that each agent $i\leq \ell$ will always be assigned to an object in these matchings. Moreover, note that 
		\begin{equation*}
		x^{\text{RSD}(I_\ell)}_{ij} = 
		\begin{cases}
		\sfrac{1}{\ell} &\text{if } j = 1, \\
		1-\sfrac{1}{\ell} &\text{if } j = 2 \text{ and } i \leq \ell, \\
		0 &\text{if } j = 2 \text{ and } i > \ell.		
		\end{cases}
		\end{equation*}
	Therefore, a feasible decomposition of $X^{\text{RSD}(I_\ell)}$ over $\mathcal{M}^{\text{SD}}_{2\ell-1}$ can be constructed by giving matching~$M^t$ a weight of $\lambda^t=\sfrac{1}{\ell}$, for each $t = 1, \ldots, \ell$. 
	
	Combining both results, we obtain that 
	\begin{equation*}
	\lim_{\ell \to \infty}\frac{z(X^{\text{RSD}(I_\ell)})}{p^-(I_\ell)} = \lim_{\ell \to \infty}\frac{2\ell - 1}{\ell} = 2, 
	\end{equation*}
	which proves the theorem.
\end{document}